\documentclass[a4paper,10pt]{IEEEtran}

\usepackage{cancel} 
\usepackage{graphics,color} 
\usepackage{subfigure}
\usepackage{rotating}
\usepackage{graphicx}
\usepackage[linesnumbered,ruled,vlined]{algorithm2e}
\usepackage{times}
\usepackage{comment}
\usepackage{amssymb,url}
\usepackage{epsfig}
\usepackage[cmex10]{amsmath}
\usepackage{amsthm}
\usepackage{bm}
\usepackage{cite}
\usepackage{breqn}
\usepackage{fancybox}
\usepackage{textcomp}

\usepackage{ragged2e}
\usepackage{amsmath}
\usepackage{amsthm}
\usepackage{amsfonts}
\usepackage{amssymb}
\usepackage{breqn}
\newtheorem{lemma}{Lemma}
\newtheorem{definition}{Definition}
\newtheorem{theorem}{Theorem}
\newtheorem{corollary}{Corollary}

\def\Prf{\vspace{2ex}\noindent{\bf Proof: }}
\def\endpf{\hfill{$\blacksquare$}}
\def\arcsinh{\mbox{arcsinh}}
\newcommand{\JOAC}{JOAC\xspace}
\newcommand{\TJOAC}{T-JOAC\xspace}

\usepackage{todonotes}
\newcounter{cori}
 
 \newcounter{mand}

 \newcommand{\shadow}{}

\usepackage[german,english]{babel} 
\usepackage{booktabs} 
\usepackage{comment} 
\usepackage[utf8]{inputenc} 
\usepackage[T1]{fontenc} 

\usepackage{amsmath,amsfonts,amsthm} 
\usepackage{tikz} 
\usepackage{tikz-cd}
\usetikzlibrary{positioning,arrows} 
\usetikzlibrary{decorations.pathreplacing} 
\usetikzlibrary{shapes, arrows} 
\usepackage{tikzsymbols} 
\usepackage{blindtext} 
\usepackage{tikz}
\usetikzlibrary{fit,tikzmark}
\begin{document}



\title{Joint Traffic Offloading and  Aging Control  in  \\ 5G IoT Networks  } 
\author{Naresh Modina\IEEEauthorrefmark{1}, Rachid El-Azouzi\IEEEauthorrefmark{1}, Francesco De Pellegrini\IEEEauthorrefmark{1}, Daniel Sadoc  Menasche\IEEEauthorrefmark{4} and Rosa Figueiredo\IEEEauthorrefmark{1}\thanks{This project was partially sponsored by CAPES, CNPq and FAPERJ, through grants  E-26/203.215/2017 and E-26/211.144/2019.
}\\
\thanks{This work was a part of the project MAnagEment of Slices in The Radio access Of 5G networks (MAESTROG)}
\thanks{An early version of this manuscript has appeared in the proceedings of the International Teletraffic Congress (ITC) 2020.}
\IEEEauthorrefmark{1}University of Avignon, France;  \IEEEauthorrefmark{4}Federal University of Rio de Janeiro, Brazil.}
\maketitle 

\def\Prf{\vspace{2ex}\noindent{\bf Proof: }}
\def\endpf{\hfill{$\blacksquare$}}
\def\arcsinh{\mbox{arcsinh}}
\newcommand {\gx} {\gamma_x}
\newcommand {\gy} {\gamma_y}
\newcommand {\T}{T_{int}}
\newcommand {\E}{\mathbb{E}}
\newcommand {\R}{\mathbb{R}}
\newcommand {{\Loc}}{\mathcal{L}}
\renewcommand{\L}{\mathcal{L}}
\newcommand {{\Tseconds}}{\kappa}
\newcommand {\thresholdvector}{{\bm \tau}}
\newcommand {\threshold}{{\tau}}
\newcommand {\thresholdrange}[1]{{\tau}^{({#1})}}
\newcommand {\thresholdloc}[1]{{\tau}_{{#1}}}

\newcommand{\M}{\mathcal M}
\newcommand{\OP}{\texttt{OptimalPricing}\xspace}
\newcommand{\taboo}[4]{ {}_{\displaystyle_{ #1}}{\!\!\!\lambda_{{#2}{#3}}^{#4}}}
\newcommand{\fdp}[1]{\textcolor{blue}{#1}}
\newcommand{\1}[1]{\mathbbm{1}\!\!\left \{#1\right \}}



\begin{abstract}
The widespread adoption of  5G cellular technology  will evolve as one of the major drivers for the growth of IoT-based applications. In this  paper, we consider a  Service Provider (SP) that launches a smart city service based on IoT data readings: in order to serve IoT data collected across different locations, the SP dynamically negotiates and rescales bandwidth and service functions. 5G network slicing functions are key to lease appropriate amount of resources over heterogeneous access technologies and different site types. Also, different infrastructure providers will charge slicing service depending on specific access technology supported across  sites and IoT data collection patterns. 

We introduce a pricing mechanism based on Age of Information (AoI)  to reduce the cost of SPs. It provides incentives for devices to smooth traffic by shifting part of the traffic load from highly congested and more expensive locations to lesser charged ones, while meeting QoS requirements of the IoT service. The proposed optimal pricing scheme comprises a two-stage decision process, where the SP determines the pricing of each location and devices schedule uploads of collected data based on the optimal uploading policy.  
Simulations show that the SP attains consistent cost reductions tuning the trade-off between slicing costs and the AoI of uploaded IoT data.    
\end{abstract}





\section{Introduction}

Data collection at scale represents the key signature of future IoT applications, posing significant challenges in the integration of emerging 5G networks and IoT technologies as identified in early studies \cite{surveyIoT}. In fact, pervasive object readings will play a decisive role in the context of smart cities for both process monitoring and management \cite{smartcity}. Using IoT, a whole new set of applications will be able to feed local information generated by both objects and mobile devices into their databases. Such information streams are consumed for management and prediction purposes by services such as city air management, smart waste management or traffic management, and demand-response schemes~\cite{zanella2014}. Data brokerage is thus emerging as one of the most interesting business opportunities: new service providers in 5G networks can seize the opportunity to mediate between companies purchasing IoT data and device owners. This is considered a cornerstone in creating a marketplace for IoT data \cite{IoTmarket2,IoTmarket3,IoTmarket4,IoTmarket5} which is essential for the uptake of smart city services. 

The architecture of IoT networks must be able to support local data streams collected from highly heterogeneous information sources, including e.g., meters for water and electricity management, outdoor and indoor positioning data, parking presence sensors, and a whole new set of user-generated contents related to mobile application-specific data. 
Indeed, the long-standing problem of integrated architectures and protocols to support IoT data collection appears finally solved by the uptake of 5G connectivity \cite{cisco2020}. Slicing techniques offered by 5G technology allow Infrastructure Providers (InP) to offer differentiated services to their customers using shared resource pools. A slice for IoT services, in this context, is a share of mobile network infrastructure obtained by forming a logical network on top of the physical one connecting IoT devices (Fig.\ref{heteroGeo}). 
More generally,  traffic differentiation in 5G systems can be obtained by isolating specific traffic categories within slices, which in turn can be dedicated to serving target verticals under specific service isolation guarantees \cite{ZhangCommMag}\cite{SamdanisMT}. %
Smart city services, where Service Providers (SP) support IoT data readings from mobile sensing devices are a key use case of \emph{slicing service}. In this context, the role of the SP is to lease resources (radio, processing, storage and  radio resources) in the form of one or more dedicated slices and from one or multiple InPs; the leased slice will support the connectivity of the fleet of devices taking part to the IoT sensing services at a cost for the upload of sensed data.  

The costs incurred by sensing services depend on a number of factors, including the business model and the ownership of the sensing devices. It is possible that sensing devices are owned by the SP whereas the sensing services are designed and run by third parties and offered, e.g.,  as a smartphone app. In this case, the sensing services involve payments to the SP~\cite{jayasumana2007virtual}. If the SP is in charge of the sensing services and also the slicing services, 
in turn, non-monetary costs --  similar to the shadow prices defined in \cite{Kelly1998} -- can be used effectively as a penalty to avoid hot-spot phenomena by deterring the upload of sensed data in congested areas. 

IoT sensing services relying on 5G technology pose their own challenges. The informative content of sensed data changes over time depending on the profile of the IoT sensing service. Information on traffic mobility, for instance, will retain its value on the timescale of the tenths of seconds, whereas temperature and pollution measurements will change in the timescale of the hours.  Clearly, managing IoT devices requires a mechanism to control information freshness, the latter being also referred to as the \emph{age of information} (AoI).  Such mechanism, known as aging control, determines when IoT readings should be uploaded to avoid stale information.

Controlling the AoI dynamics of data carried by IoT devices permits to trade-off between the value of IoT sensing readings -- indeed specific to a tagged service -- for the cost for uploading them using the 5G IoT slicing service. Motivated by the aging control problem intrinsic to IoT devices, and by the traffic offloading capabilities enabled by 5G technology, in this paper we investigate the following two questions: 
\begin{enumerate}
\item given the requirements of a tagged \emph{IoT sensing service} and the SP charging rates, what is the optimal upload strategy to control information freshness at the device level? 
\item how should the SP incentivize users to offload IoT data in order to reduce the costs to lease the resource slice?
\end{enumerate}
In the first part of the paper, we address the first question via the control of AoI at the device level, and we derive an optimal upload strategy. Sensing devices trigger the upload of sensed data depending on two factors: the application profile and the price for the IoT sensing service. It is the application profile to determine for how long sensed data retains their value, whereas location-dependent prices determine the unit cost of sensed data uploads performed using the IoT slice. The problem is formulated as a Markov decision process (MDP). The optimal stationary policy solving the problem has the multi-threshold structure: the upload of information occurs depending on the upload prices available to a tagged device, i.e., prices available in the cell it is connected to, and on the AoI relative to the data stored in the device memory. In the second part of the paper, the minimization of the slicing service costs is addressed: the SP optimizes the vector of prices that are exposed to devices with the aim to minimize the cost paid to the InP for leasing the slice while satisfying the applications'  delay target. 

{\em Prior art and main contribution.} The two control actions considered in this work are traffic offloading control and aging control. Traffic offloading is a standard networking technique to perform load balancing and avoid traffic congestion. However, in 5G networks, it must work on a per slice basis, and must be made available to SPs in a transparent fashion with respect to InP traffic management tools. Aging control on the other hand is a key requirement for sensing applications in IoT systems. These two problems have been addressed separately ~\cite{altman2019forever, liu2020joint, song2019age,Raja2020,Kortoci2019}, but no prior work has considered the two problems at once to the best of the authors' knowledge. This paper aims to connect these two research lines within the same control framework, resulting in a scheme for the cost-efficient brokerage of IoT data using 5G slicing. While IoT data offloading techniques have been proposed in the context of vehicular networks \cite{Raja2020} or sensor networks \cite{Kortoci2019}, the proposed solution is tailored specifically to the case of  5G slicing since the SP can stimulate IoT data offload towards less congested areas using a distributed and location-aware scheme which operates at the sensing application level. Furthermore,  by means of flexible pricing control, we minimize the cost incurred by the SP in order to lease slice resources from InPs. Finally, the proposed framework includes inherently a notion of service level agreement (SLA) since it is rooted in the concept of AoI which captures latency requirements of IoT data readings, as agreed by the SP with his customers. 

{\em Structure of the paper. } The remainder of this work is organized as follows. First we introduce the considered system in Sec.~\ref{sec2}. Then, we discuss the two main control actions, namely, traffic offloading control and aging control in Sec.~\ref{sec3} and  in Sec.~\ref{sec4}, respectively.  Sec.~\ref{sec5} bridges the two pillars in a unified framework. We propose algorithms to solve the joint offloading and AoI control in Sec.~\ref{sec:algorithms}. We report numerical results in Sec.~\ref{sec6} and we revise related works in Sec.~\ref{sec7}. A closing section ends the paper.


\section{System description}\label{sec2}	

\renewcommand{\thefootnote}{\fnsymbol{footnote}}    
	\begin{figure}[t!]
	\centering 
	\includegraphics[width=0.45\textwidth]{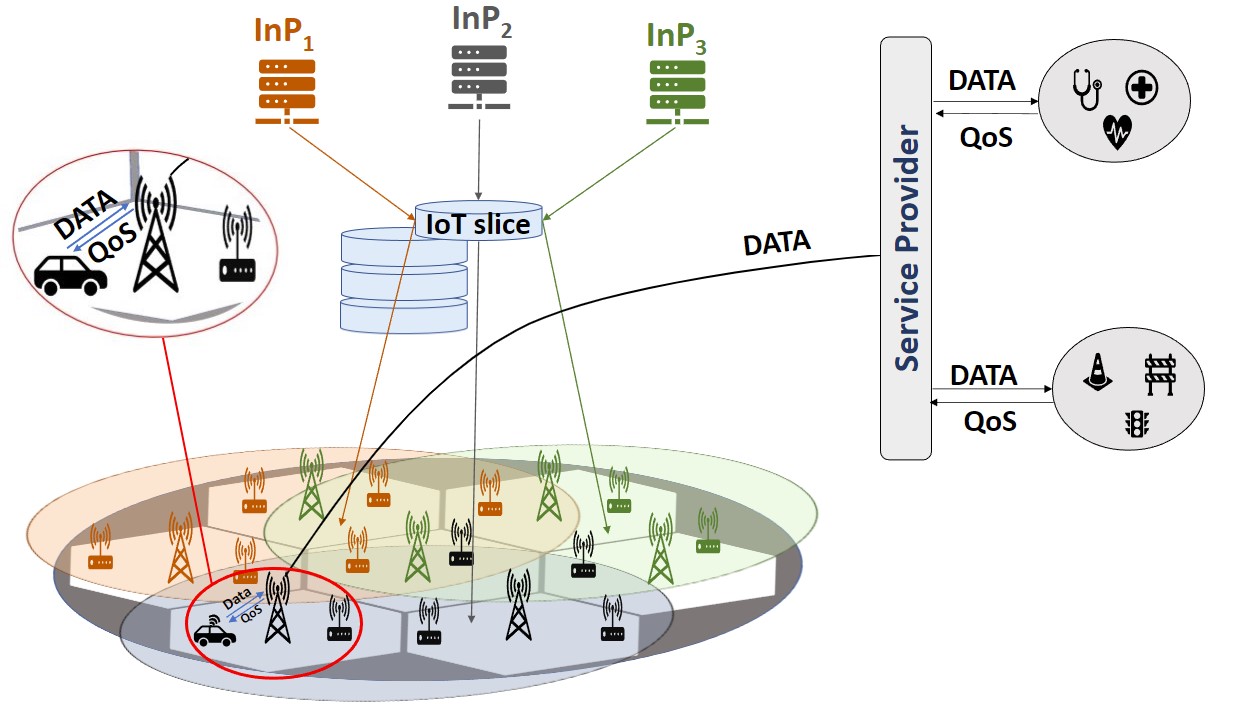}
	\caption{Cutting slices of resources across multiple locations and across multiple InPs.}
	\label{heteroGeo}
    \end{figure}  
  
A {\em Service Provider (SP)} offers Internet connectivity to heterogeneous IoT devices over a physical region (Fig.\ref{heteroGeo}). The SP can act as a data broker, i.e., it collects data from device owners or from mobile IoT devices deployed across the region and sell the data to interested parties under Data as a Service (DaaS) scheme or use the collected data to run his own service. To that aim,  a single SP can aggregate resources leased from various available InPs at different locations. Each InP provides dedicated 5G slices for IoT data collection at certain cost. In practice, sensed data is relayed using a fleet of mobile devices uploading them at the need while mobile relays are served through resources across a pool of base stations covered by the selected InPs infrastructures. 

Because sensed data belongs to a variety of categories, e.g., healthcare data, environmental monitoring data, road traffic data, etc., it has different time sensitivity. SP customers will require brokered IoT data to comply with certain QoS requirements. Throughout this work, the latency of delivered IoT data is the reference SLA metric (indeed it is a fundamental parameter for, e.g., industrial automation, intelligent transport systems, and healthcare monitoring applications). Latency, in turn, is impacted by the locations from which mobile devices upload sensed data. Note that aggregated traffic may vary significantly across regional locations, e.g., due to the presence of hotspots. Ultimately, the SPs need to grant target QoS figures for a given IoT application and obey standardized SLA. To this aim, the key enabler is 5G network slicing by which the SP negotiates and adjusts the scale of bandwidth and service functions. In practice, this entails orchestrating  slicing functionality across heterogeneous access technologies (5G, LTE, 3G, and WI-FI), over different site types (macro, micro, and pico base stations) and over multiple InPs. The cost of leased infrastructures depends on chosen InPs, specific access technology supported across regional sites, and IoT data collection patterns. For the sake of clarity, we shall refer to bandwidth costs only, but the whole framework may well include also costs for local computation and/or storage \cite{3GPPPSlices}.

In order to comply with SLA agreements for IoT data collection, the SP dynamically determines the resources per slice required to match the current demand. Due to scarcity of resources, higher costs will be incurred in crowded and congested locations. Hence, the SP designs an {\em IoT data collection policy} by which freshness of IoT data is traded off  against costs. In fact, the upload of non-critical data can be deferred to occur at a location with smaller costs yet complying with SLAs, i.e., target latency figures.  
 
The key mechanism detailed in the next section is a price-based load balancing scheme where the SP incentivizes users not to upload data from congested locations. Prices are dynamically set, e.g., based on the congestion levels. Different locations are tagged by a price to upload a unit of IoT data. The whole scheme takes advantage of user mobility: while IoT devices are carried by users appliances and move across the regions, data upload can be diverted towards less congested locations. Through pricing, the SP can shift the IoT traffic generated by mobile IoT devices from locations where data are sensed to less crowded ones, where leased slice resources are relatively cheaper. Following example provides the intuition behind the traffic offloading mechanism.

    

\begin{figure*}[t]
\centering
\includegraphics[width=0.45\textwidth]{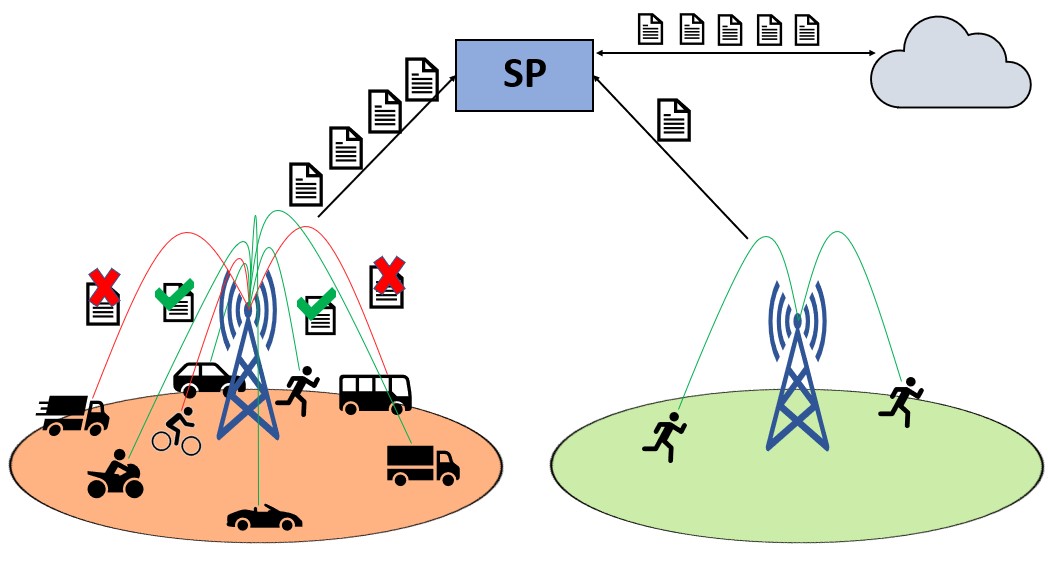} 
	\includegraphics[width=0.45\textwidth]{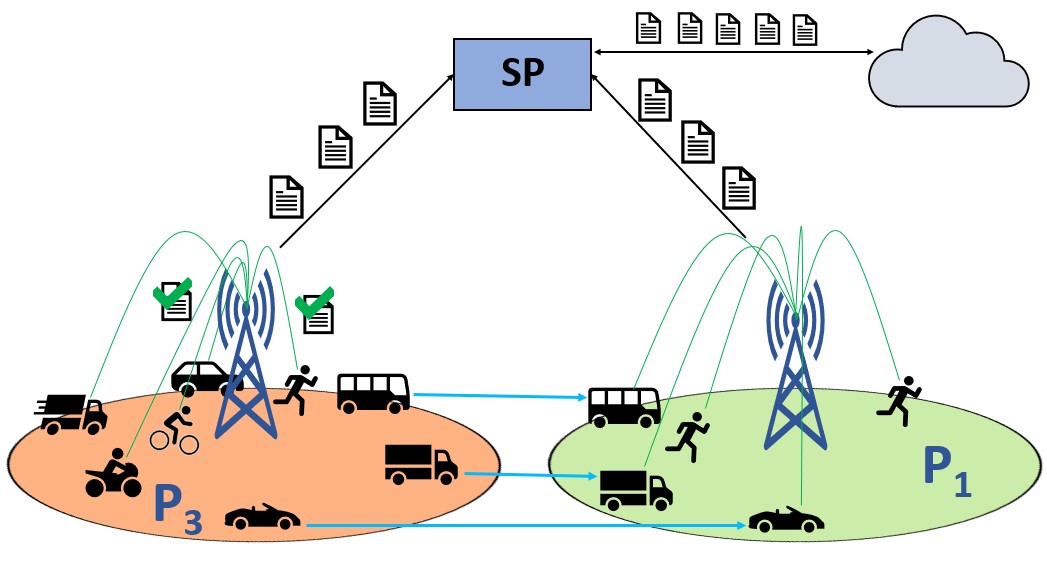}
	\put(-435,0){location $i$}
	\put(-305,0){location $j$}
	\put(-192,0){location $i$}
	\put(-70,0){location $j$}
	\caption{a) Without traffic offloading mechanism b) With traffic offloading mechanism. }
	\label{Tmigration}
\end{figure*}

Fig.~\ref{Tmigration}(a) shows a scenario that has no control mechanism for traffic upload.  As a result, up-link traffic is high in location $i$ and SP may lose part of the data due to outage.  In addition,  upload costs would be higher as well.  In Fig.~\ref{Tmigration}(b), in contrast, upload traffic is well distributed.  A traffic offloading mechanism is adopted, and there is no overburden on any particular access point.

In the following sections, we introduce the mechanism that combines aging control and traffic offloading.  This mechanism, deployed at each device, is designed based on the aggregated mobility of the devices.  Nonetheless, if the individual mobility of devices is available, the mechanism can also leverage this information, noting that the objective of the device is to find an optimal  strategy to upload the collected data.   If decisions are made per device, using individual  mobility patterns, heterogeneity across devices  does not impact the upload decisions of each device.

  
\section{ Traffic Offloading}\label{sec:offload}\label{sec3}


In this section, we formulate the SP \shadow pricing scheme which is used to minimize the total cost incurred by the SP. The resulting control problem accounts for the mobility pattern of devices and the delay requirements of the IoT data collection service. At each location, the SP  selects the corresponding InPs.  Let ${\L}=\{1, 2,..,L\}$
be the set of locations.  We denote by  $B_i$ the maximum bandwidth at location $i$, $1 \leq i \leq L$, resulting in a maximum bandwidth vector  $\bm{B}$. Each location is tagged with a unit \shadow price, resulting in a \shadow price vector ${\bm p}=(p(i), i\in{\L})\in\mathbb{R}^L$. 
The price vector induces a set of $K$ different prices denoted by   ${\mathcal{P}}_K$, $\mathcal{P}_K=\{P_1,..,P_K\}$, where $P_1< \ldots <P_K$. Prices  impact  location-dependent upload policies which determine when a mobile IoT device should upload sensed data, based on the current age of information. The age of information represents the time elapsed from sensor reading until upload. Let $\Delta_i(p)$ be the random variable representing the  age of information based on the price vector -- at upload time --  for data collected at location $i$. Let $C_i$ be the monetary unit cost to lease bandwidth at location $i$. Let $D_i$ be the amount of data generated by devices at location $i$ during a time slot. Finally, the upload control is represented by variable $Y_{ij}({\bm p})$,  which is the average traffic rate for data collected at location $i$ and uploaded at location $j$. Notations are summarized in Table~\ref{tab:notation}.
  
\label{sec:P1}	 Next, we pose the  optimization problem faced by the SP, named here as\\
\textsc{Traffic Offloading: } 
\begin{align}
 &\underset{{\bm{p}}}{\text{minimize}}  \sum_{i\in {\Loc}} \sum_{j\in {\Loc}}  C_j \, Y_{ij}({\bm p}) \nonumber \\
& \text{subject to}\nonumber\\
&   \hskip 5mm  \sum_{j\in {\Loc}}  Y_{ji}({\bm p}) \leq B_i  ,\; \forall i \in {\Loc} \label{eq:c1}\\
&   \hskip 5mm  \sum_{j\in {\Loc}}  Y_{ij}({\bm p}) = D_i,\; \forall i \in {\Loc}\label{eq:c2}\\
&   \hskip 5mm  \mathbb{P} ( \Delta_i({\bm p}) > d) \leq \epsilon,\; \forall i \in {\Loc}   \label{eq:c3}\\
&  \hskip 5mm  Y_{ij} \geq 0,\;  \forall i, \forall j \in {\Loc}
\end{align}
In this model, \eqref{eq:c1} is the per location constraint on the available bandwidth for the IoT slice, and \eqref{eq:c2} is a flow conservation constraint. 
Constraint \eqref{eq:c3} provides a tunable SLA constraint on the age of information collected at specific location $i \in {\Loc}$, depending on a target latency value $ d> 0$ and on tolerance $\epsilon>0$. 
%
%

The main challenge  to solve the \textsc{Traffic Offloading} problem is to account for the mobility pattern of devices. In fact, they collect data at some tagged location, and they upload it according to the chosen policy, in order to meet QoS requirements. In practice, once a sensing device is associated with a tagged location, it will be informed of a \shadow price available for the IoT slicing service, so that the decision to upload or not can be implemented onboard of sensing devices in a fully distributed fashion. In Section \ref{sec5} we shall provide an algorithm able to determine the optimal price vector ${\bm p}$ solving the traffic offloading problem. Before that, in the next section, we determine the optimal upload control at the device for a given price vector ${\bm p}$. 

AoI is key for IoT data, as it represents how ``fresh'' is information, e.g., given the rate at which new samples are produced by a given  sensor~\cite{AOI1,AOI2,AOI3}  AoI is the time difference between the time of generation and the time of successful delivery of data to destination.  In general different services have different QoS requirements: services using data collected from various IoT devices have thus differnt tolerance with respect to the age of information. For example services which provide traffic status updates may need AoI in the order of seconds. Conversely, traffic status predictors may tolerate higher values of AoI.  In essence, AoI is a new performance metric able to complement traditional throughput/delay-based performance evaluation.

Throughout this work, the control of AoI is corresponds to the upload control with the upload control. More explicitly, the reward that the MSP achieves by uploading data depends on the value of the information at upload time.  The case for a utility-based model is that it is flexible enough to allow the MSP to run different services by simply changing the utility function and pricing scheme. As an example, since the value of the information for real time updates decays very fast, it could be approximated as an exponential function. Different AoI utility functions could possibly used to match different QoS requirements \cite{AOI4}.


\begin{table}[t]\caption{Table of notation}
\centering
\begin{tabular}{p{0.18\columnwidth}|p{0.7\columnwidth}}
\hline
{\it Notation} & {\it Description}\\
\hline
\multicolumn{2}{c}{Basic parameters} \\
\hline
${\Loc}$ & set of regional locations; ${\Loc}=\{1,\ldots,L\}$\\
$\mathcal P$  & set of upload unit prices; $\mathcal P=\{p_1,\ldots,p_K\}$\\
$\mathcal M$  & set of information age values; $\mathcal M=\{1,\ldots,M\}$\\
$D_i$& amount of data generated at location $i$ during a time slot\\
$\Pi_{ij}$& transition probability from location $i$ to  $j$ \\
$\taboo{A\;\;}{i}{j}{n}$& probability of moving from location $i$ to   $j$ in $n$ steps without entering   taboo set $A$\\
$\pi_{i}$& occupation probability of location $i$ \\
$B_i$& maximum bandwidth at location $i$\\
$\textbf{B}$& Maximum bandwidth vector \\
$C_i$& monetary unit cost to lease bandwidth at location $i$\\
$d$& target latency\\
$\epsilon_i$& tolerance factor at location  $i$\\
$N$& number of IoT devices\\
$F$& average size of the collected data\\
$\Tseconds$&  timeslot duration (seconds) \\
\hline
\multicolumn{2}{c}{States, actions, transitions and rewards }\\
\hline
$t$ & current timeslot \\
$x_t$& age of information at time $t$  \\
$l_t$& location at time $t$ \\
$s_{t}$& state at time $t$; $s_{t}=(x_t,l_t)$\\
$a_{t}$& action at time $t$, where 1 means upload, and 0 defer \\
$\mu(x,l)$& function expressing the probability that the device performs action $a = 1$ in state~$s=(x,l)$\\
$\Gamma_{s,a,s'}$& transition probability from $s$ to $s'$ under action~$a$\\
$r_t(s_t, a_t)$& instantaneous reward under state action pair $(s_t,a_t)$ at time~$t$\\
\hline
\multicolumn{2}{c}{Variables}\\
\hline
$\Delta_i({\bm p})$& random variable characterizing     age of information at upload time for data collected at location $i$\\ 
$F_{\Delta_i}$& probability distribution (CCDF) for the age of information at upload time for data collected at location $i$\\
$Y_{ij}({\bm p})$& average traffic rate for data collected at location $i$ and uploaded at location~$j$\\
$y_{ij}$& 		   average traffic rate for data collected at location $i$ and uploaded at location~$j$, per device\\
$f(i,j,t; {\bm p})$& probability that a device collects data from location $i$ and upload it at time $t$ in location $j$\\
$\mathcal{U}$& set of prices corresponding to locations wherein the optimal policy is to upload\\
$K$ & number of threshold values in the current multi-threshold policy \\
$\thresholdrange{j}$& $j^{th}$ AoI  threshold value, $\thresholdrange{1}=0$, $\thresholdrange{j} \leq \thresholdrange{j+1}$,   and $\thresholdrange{K} \leq M$ (for convenience,  $\thresholdrange{K+1} = M$)  \\
$\threshold_l$&   AoI threshold for data collected at location~$l$\\
$\thresholdvector$&  AoI  threshold vector (one threshold per  location); $\thresholdvector = (\threshold_1,  \ldots, \threshold_L)$\\
$\thresholdvector_{max}$& AoI  threshold vector with all  values equal $\threshold_{max}$ (maximum achievable AoI)   \\
\hline
\end{tabular}\label{tab:notation}
\end{table}


\section{ Aging Control }\label{sec4}	




Each device decides to upload data or defer based on its actual location, the vector of \shadow prices, and the age of information stored in its buffer. Let $x_t$  be the age of information for data collected at time $t$ by a tagged device: $x_t=1$ when the device collects it, and increases by one at every time slot, except when the device uploads data or the collected data reaches the maximum age, denoted by $M$. Note that $M$ is a design parameter, assumed to be fixed and given.   
We let ${\mathcal{M}}=\{1,\ldots,M\}$ so that $x_t \in \mathcal{M}$. 
Let $U(x)$ be the utility corresponding to uploading data with age of information $x$, where $U(\cdot)$ is a non-increasing function. The selection of the utility function is up to the SP.  For example,  if the SP wants to collect data concerning traffic updates, the value of the information may decrease exponentially fast.  For pollution level updates, in contrast, the value may not decrease as fast.

The state of a tagged device at time $t$ is denoted $s_t=(x_t,l_t)$, where $x_t$ is the AoI  as described in the above paragraph and $l_t$ is the device's location at time $t$. Location $l_t\in {\mathcal L}$ is the state of a finite, discrete, ergodic Markov chain, whose dynamics determines the mobility pattern. We denote the transition probability between location $l$ and $k$ by $\lambda_{lk}$; $\Lambda=\{\lambda_{lk}\}$, is the corresponding transition probability matrix. Finally, let ${\bm \pi}=[\pi_1, \pi_2, .. \pi_L]$ be the steady state probability distribution. 

The action set available at each device is to upload or defer, i.e., $A=\{0,1\}$, where $0$ means ``defer'' and $1$ ``upload'';  the action taken at time $t$ is denoted by $a_t$. Hence the dynamics of the age of information  at a tagged device is   given by
\begin{equation}
x_{t+1}= 
\begin{cases} 
\;1,& \quad \mbox{if} \; a_t = 1, \\
\;\min(x_t+1, M), & \quad \mbox{if} \; a_t = 0. \\
\end{cases}\nonumber
\end{equation}
Next, we characterize the transition probability of the resulting MDP. Let $s=(x, l)$ be the current state of the device and let  $s'=(x', l')$ be its next state under action $a$. The transition probability from $s$ to $s'$, under action $a$, is given by 
\begin{equation}
\Gamma_{s , a, s'}=  
\begin{cases}
\lambda_{l,l'}, & \mbox{ if }  x'=\min(x+1,M)  \mbox{ and } a=0 \\
              & \mbox{ or } x'=1  \mbox{ and } a=1,\\
0, & \mbox{ otherwise.}  
\end{cases}
\end{equation}

{\bf Instantaneous reward}. The instantaneous reward under the state action pair $(s_t,a_t)$ at  time $t$, $r_t(s_t,a_t)$, is  
\begin{equation}
r_t(s_t,a_t) =  U(x_t) -  p(l_t) \cdot a_t.
\end{equation}


{\bf Upload policy}. The upload policy $\mu$ for a tagged device is a probability distribution over the action space. In the rest of the discussion, we restrict to stationary policies; since our action space is a binary set, a policy simplifies into function $\mu = \mu(s)$ expressing the probability the device performs action $a=1$ in state $s$.

{\bf Problem statement:} The objective of each device is to maximize the expected average reward: 
\begin{align}\label{eq:avcost}
\textsc{Aging control: } & \max_\mu \E[r,\mu] \\
&     \E[r,\mu] = \lim_{\eta \to \infty}\frac{1}{\eta}\sum_{t=0}^{\eta-1}\E[r_t(x_t,l_t,a_t);\mu] \nonumber
\end{align}
Note that, if $P_1=0$,  for any optimal strategy, the devices will upload  immediately their collected data  at locations with the price $P_1=0$.   Alternatively,  if $P_1>0$,  the  instantaneous reward can be expressed as 
\begin{equation}
r_t(s_t,a_t) =  U(x_t) - (p(l_t)-P_1)  \cdot a_t -P_1 a_t \label{eq:rewr}
\end{equation}
Therefore, the value $P_1$ can  be interpreted as the energy cost of each uploaded message.  In the remainder of this paper, except otherwise noted, and without loss of generality, we assume $P_1 =0$.

In what follows we characterize the optimal control policy that solves \eqref{eq:avcost}. We begin by introducing a special type of strategy, referred to as a {\em multi-threshold strategy}. 

\begin{definition}[Multi-threshold strategy]
A multi-threshold strategy is such that there exists $K$ and threshold values $\thresholdrange{j}$, $j=0,\ldots,K-1$ such that $\thresholdrange{1} \leq \thresholdrange{2} \leq \ldots \leq \thresholdrange{K} \leq M$ and  
\[
\mu(x, l) = 
\begin{cases}
1   & \mbox{if}\quad  x \ge \thresholdrange{j} \mbox{ and } p(l) \leq P_j \\
0   & \mbox{otherwise}
\end{cases}
\]
Note that $K$ is the number of thresholds, and $ \thresholdrange{1}$ and $\thresholdrange{K}$ are the minimum and maximum threshold values,
 $\thresholdrange{1}=0$ and  $\thresholdrange{K} \leq M$.     
\end{definition}
A device using this multi-threshold strategy uploads the collected data at its current location $l$ at price $p(l)=P_j$ if the age of  information  exceeds  $\thresholdrange{j-1}$. The following theorem reduces the problem of finding the optimal strategy  
for the \textsc{Aging control} problem to the one of
finding the $K$  thresholds  $\thresholdrange{j}$, $j=1, \ldots, K$.
\begin{theorem}
\label{THRESHOLD}
The optimization problem \eqref{eq:avcost} admits a unique deterministic optimal multi-threshold strategy.
\end{theorem}
The proof of the above theorem is available in Appendix.~\ref{ap:prooftheor1}.

 \begin{figure}[t!]
	\centering 
	\includegraphics[width=0.45\textwidth]{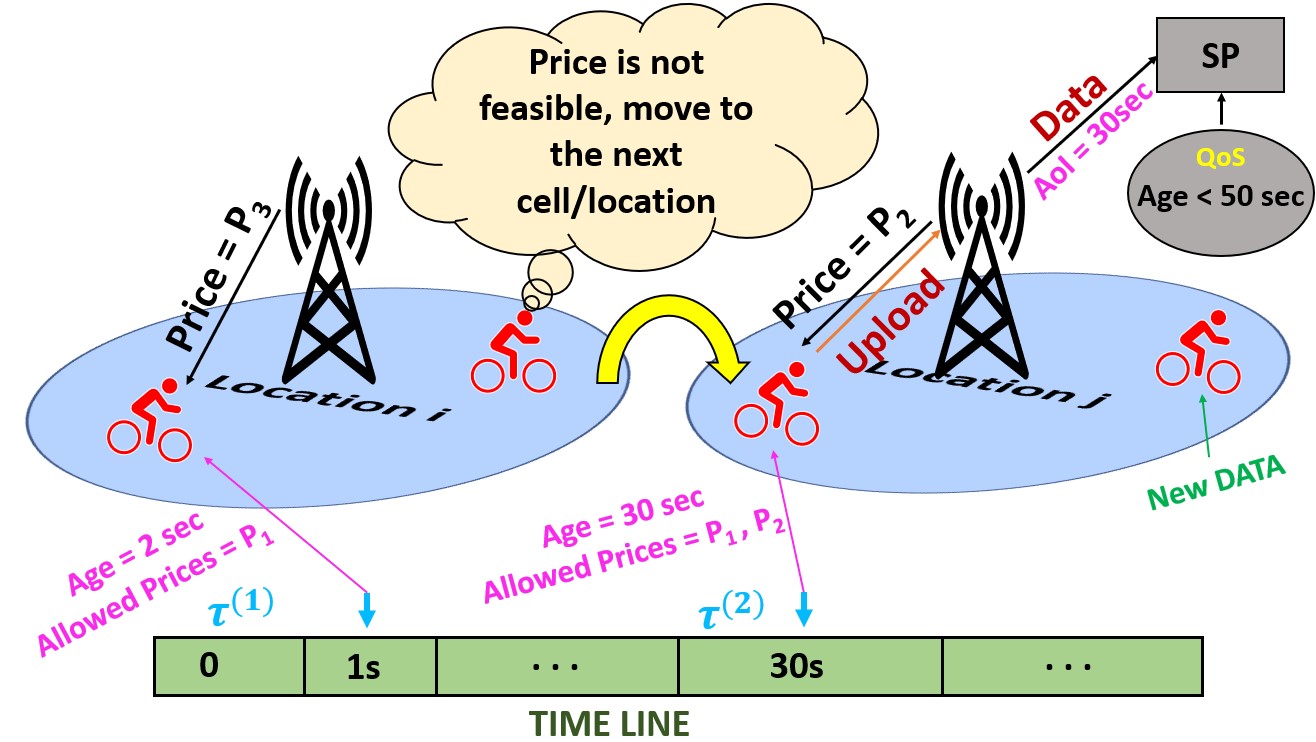}
	\caption{The upload mechanism: depending on the AoI, the upload decision is taken based on the shadow price value at the current location.} 
	\label{upload}
    \end{figure}
\begin{figure}[t]
	\centering 
	\includegraphics[scale=0.28]{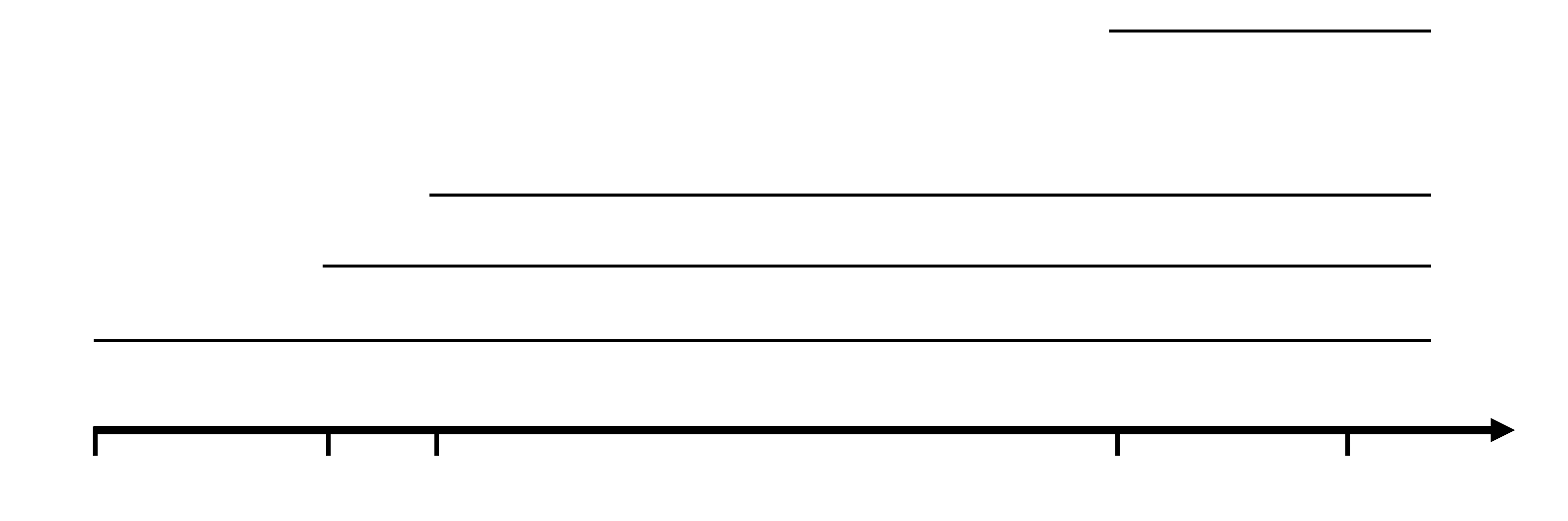}
	\put(-244,3){\put(0,0){$0$}\put(34,0){$\thresholdrange{2}$}\put(54,0){$\thresholdrange{3}$}\put(100,0){$\ldots$}\put(165,0){$\thresholdrange{{K}}$}\put(202,0){$M$}\put(224,5){$x$}}
	\put(-246,33){\put(0,0){$P_1$}\put(36,13){$P_2$}\put(56,24){$P_3$}\put(160,28){$\vdots$}\put(165,52){$P_{_{K}}$}}
	\caption{Structure of the multi-threshold policy; at the increase of the age of information upload action is optimal for an increasingly larger set of prices.} 
	\label{fig:multithr}
    \end{figure}
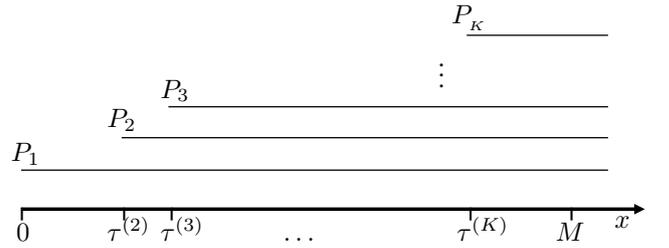

We now characterize some further properties of the optimal thresholds.  A qualitative description of the behavior of the optimal policy is depicted in Fig.~\ref{fig:multithr}. We observe that a multi-threshold strategy is a  simple procedure  to implement the distributed IoT upload control. In practice, when data is stored on a device, the AoI is one. Thus the device at the beginning will start by uploading only at locations where the price is $P_1=0$. Whenever AoI reaches $\thresholdrange{2}$, i.e.,  $x \ge \thresholdrange{2}$, the device switches to a second phase wherein an upload occurs if prices are less than or equal to $P_2$, that is either in locations with corresponding prices $P_1$ or  $P_2$. Similarly, once a new threshold is reached, say $\thresholdrange{j}$, the device will   upload the collected data at  locations with a price  less than or equal to $P_{j}$.

\textbf{Illustrative example. } Figure \ref{upload} displays a simple illustration of the multi-threshold policy, wherein a device attached to a bicycle enters a location $i$ where the price is $P_3$ and the age of the information is such that it can only upload if  price is $P_1$. After displacement,  the device enters a new location $j$ with a tagged price of $P_2$. By this time, age is higher than threshold $\thresholdrange{2}$,  allowing the device to upload information with either $P_1$ or $P_2$. Now, the device can upload the information as the price is acceptable.

As an immediate consequence of the proof of the previous theorem, we obtain the following corollary.   
\begin{corollary}
At any location $l$, $\mu(x,l)=1$ if $x\geq \thresholdrange{K}$. 
\end{corollary}
The above corollary  implies that  the maximum age that can be reached by a message is $\thresholdrange{K}$, where $\thresholdrange{K} \le M$.

In general, the set of locations where a device is allowed to upload data, as well as the age of information when the upload action is performed,  depends on the distribution of the prices  across the set of locations ${\Loc}$ used for the slice leased by the SP. Such distribution can be optimized to reduce the cost of infrastructure utilization and yet satisfy the QoS requirements of the IoT service. In the next section, we shall connect the dynamics of AoI, the structure of the multi-threshold strategy, and the distribution of the  prices. Before doing so, we shall further characterize additional properties of the multi-threshold policy. In particular, a key step is to characterize the number of prices that the optimal  threshold strategy uses with positive probability.   

Let ${\mathcal L}_i =\{l\in {\mathcal L} | p(l) \leq P_i\}$ and   $K_{lP_{i}}=\sum_{l'\in {\mathcal L}_i} \lambda_{ll'}$. In addition, 
\begin{equation}
\mathcal{S}(i)=\sum_{x=2}^{M} (U(x)-U(M)) (1-K_{lP_{i}}) + U(1)-U(M)
\end{equation}

and
\begin{equation}
\bar K_{lP_{i}} = K_{lP_{i}} - K_{lP_{i-1}}.
\end{equation}

\begin{theorem}\label{PR-THRE}
Let $\mathcal{U}$ be the set of prices corresponding to locations wherein  the optimal policy is to upload.   Then, 
%
\begin{itemize}
\item   $\mathcal{U} = \{P_1\}$,  if and only if 
\begin{equation}
\mathcal{S}(1) < p(l),\;\; \forall  l\in {\mathcal L}/{\mathcal L}_1
\label{p1}
\end{equation}
\item  $\{P_i \} \subseteq \mathcal{U}$,  with $P_i\not=0$,  if and only if 
\begin{align}
&U(1) - (\bar K_{lP_{i}} U(2)+ (1-\bar K_{lP_{i}}) U(M)) > P_i \label{ineq:cond1} \\
& \mathcal{S}(i)< p(l),  \quad \forall  l\in {\mathcal L}/{\mathcal L}_i,
\label{pi}
\end{align}
\item   $\{P_i, P_{i+1},..,P_{i+k}\}  \subseteq \mathcal{U}$,  if and only if condition~\eqref{ineq:cond1} is met and  
\begin{align}
&\mathcal{S}(i+k)< p(l), \quad \forall  l\in {\mathcal L}/{\mathcal L}_{i+k}.
\end{align}
\end{itemize}
\end{theorem}
The proof  of the above theorem is available in Appendix.~\ref{app:propo}.  

Theorem~\ref{PR-THRE}  establishes  conditions under which devices upload data if and only if they are found in a given finite set of locations.  In the following section, we derive optimal pricing assignments minimizing SP costs while still satisfying users QoS requirements, for a given assignment of price to locations.

\begin{corollary}
If $P_1 \neq 0$, $\mathcal{U} = \emptyset$,  if and only if 
\begin{equation}
\sum_{x = 1}^{M}(U(x)-U(M)) < p(l),\;\; \forall  l\in {\mathcal L},
\label{c1}
\end{equation}
\label{c2}
\end{corollary}

Proof of the above corollary is available in Appendix ~\ref{pr:corro2}



\section{Joint  {Aging} Control and  Traffic Offloading}~\label{sec5}



Once we determined the optimal distributed upload control, we return to the \textsc{Traffic Offloading}  problem introduced in Sec.~\ref{sec:P1}. 
%
\subsection{Pricing as a tool for joint aging control and offloading }

We recall that the SP aims at setting optimally the value of the shadow prices to reduce the total cost to lease resources from different InPs. Let us assume $N$ IoT devices spread over the set of locations ${\Loc}$. Each device generates data to be collected and sent to the IoT server located in the core network every $\Tseconds$ seconds.  Let $\pi_j$ be the ergodic probability of a device collecting data at location $j$ -- which in turn depends on the mobility profile of devices. Hence, the total rate of collected  data  by devices in location $j$ is given by
\begin{equation}
D_j ={N \pi_j F}/{\Tseconds},
\end{equation}
where $F$ is the average size of the collected data.  

First, observe that if shadow prices are constant over locations, i.e., $p(l)=P_1$, for $l\in {\Loc}$, each device will transmit immediately the collected data and  the total cost for SP is
\begin{equation}
\sum_{j\in {\mathcal L}} C_j D_j=  \frac{NF}{\Tseconds}\sum_{j\in {\mathcal L}}  C_j \pi_j.
\end{equation}
Our primary interest for the distributed upload control via shadow pricing is to perform load balancing by shifting part of the traffic load from highly congested locations, which we expect indeed to be more expensive to lease, to lesser charged  locations. At the same time, we aim at ensuring the QoS requirements of the IoT service. Under shadow pricing  vector $\bm p$, the total rate uploaded at location $j$  under the optimal threshold strategy 
is given by 
\begin{equation}
Y_j({\bm p}) = \sum_{i\in {\mathcal L}} Y_{ij}({\bm p}) = \sum_{i\in {\mathcal L}} D_{i} y_{ij}({\bm p})  =  \frac{NF}{\Tseconds}\sum_{i\in {\mathcal L}}  \pi_i y_{ij}({\bm p}).
\label{rateu}
\end{equation}
Hence the total cost writes 
\begin{equation}
\sum_{j\in {\mathcal L}} C_j Y_j({\bm p})=  \frac{NF}{\Tseconds}\sum_{j\in {\mathcal L}} C_j\left(\sum_{i\in {\mathcal L}} \pi_i y_{ij}({\bm p}) \right).\label{eq:cy}
\end{equation}
In what follows, we leverage the above equation as the objective of our optimization problem.

\subsection{Formulation of joint offloading and aging control problem}

Next, we  account for the  \textsc{Aging Control} problem introduced in Sec.~\ref{sec4}
 under the   \textsc{Traffic Offloading}  problem introduced in Sec.~\ref{sec3}. The resulting
joint problem is posed as follows. 
\begin{align}
 & \textsc{Joint  Offloading and Aging Control (JOAC)}: \nonumber \\
 &\hspace{10pt} \underset{{\bm p}}{\text{minimize}} \sum_{i\in {\mathcal L}}\pi_{i} \sum_{j\in {{\Loc}}}    y_{ij}({\bm p}) C_ j \label{C0}
 \end{align}
subject to
\begin{align}
& \hspace{10pt} \hspace{10pt}  \sum_{i \in {{\Loc}}}  \pi_i y_{ij}({\bm p}) \leq f(B_j),\; \forall j \in {\mathcal L} \label{C1}\\
&  \hspace{10pt} \hspace{10pt}  \sum_{j \in {\mathcal L}}  y_{ij}({\bm p}) = D_i,\; \forall i \in {\mathcal L} \label{C1b} \\
& \hspace{10pt} \hspace{10pt}  \mathbb{P}( \Delta_i({\bm p}) > d) \leq \epsilon,\; \forall i \in {\mathcal L}\label{C2}\\
& \hspace{10pt} \hspace{10pt} y_{ij}({\bm p}) \geq 0,\;  \forall i, \forall j \in {\mathcal L}\label{C2b}
\end{align}
where $y_{ij}$ is the expected per device upload rate for data collected at location $i$ and uploaded at location $j$.   

We can calculate $y_{ij}({\bm p})$ based on the threshold strategy from section \ref{sec4}: we need to calculate the probability that a device collects data at location $i$ and uploads it at  location $j$. The calculation is performed by determining $f(i, j, t)$, namely the probability that a device collects data from location $i$ and uploads it at time $t$ in location  $j$. Such computation involves the use of taboo probability, defined as follows:
$$ 
\taboo{A\;\;}{i}{j}{n} = \mathbb{P}\left(l_1,..,l_{n-1} \not\in A, l_{n}=j  |  l_0=i \right). 
$$
This is the  probability   of moving from location $i$ to location $j$ in $n$ steps without entering the taboo set  $A$; such transition probabilities are calculated in the standard way by considering the $n$-th power of the taboo matrix, which is obtained by zeroing the columns and the rows of the transition probability matrix corresponding to the taboo states, i.e., the states in $A$. 
Based on the optimal threshold strategy, if a device collects data from a location $i\in {\mathcal L}_1$, it will immediately upload it.  Thus for  $i\in {\mathcal L}_1$, we have 
$$
f(i, z , t; {\bm p})= \begin{cases}
1,   & \mbox{  if }   z=i \mbox{ and } t=1,\\
0,  & \mbox{otherwise.}
\end{cases}
$$
For $i\not\in {{\Loc}}_1$ and $z\in {\mathcal L}_j$, let us consider $\widehat{\threshold}^{(t)}=\max(\thresholdrange{j}|\thresholdrange{j}<t)$. The explicit expression can be derived as follows
\begin{eqnarray}\label{eq:fizt}
&&\hskip-11mm f(i, z, t; {\bm p})=0,  \mbox{  for }  t  < \thresholdrange{j}\\
&&\hskip-11mm f(i, z , t; {\bm p})=\nonumber\\
&&\hskip-11mm \sum_{l_1\not\in {\Loc}_1}\sum_{l_2\not\in {\Loc}_2} \dots \!\!\!\sum_{l_{t-1}\not\in {\Loc}_{\widehat{\threshold}^{(t)}}}  \taboo{{\Loc}_1\;}{i}{l_1}{n_1}\cdot \taboo{{\Loc}_2\;}{l_1}{l_2}{n_2}\ldots \taboo{{\Loc}_{\widehat{\threshold}^{(t)}}}{l_{t-2}}{l_{t-1}}{t-{\widehat{\threshold}^{(t)}}-1}\lambda_{l_{t-1} j}\nonumber\\
&&\hskip40mm  \mbox{  for }   \thresholdrange{j} \leq t  \leq \thresholdrange{K} \\
&&\hskip-11mm f(i, z, t)=0,  
\mbox{  for }  \thresholdrange{K} <  t  \leq M
\end{eqnarray}
where 
$$n_k= \thresholdrange{k+1} - \thresholdrange{k}, \quad  k=1,\ldots,K$$
 with $\thresholdrange{K+1}=M$.  The expression of $y_{iz}$ for $z \in {{\Loc}}_j$  yields 
\begin{equation}\label{eq:yiz}
y_{iz} =\sum_{t=\thresholdrange{j}}^{\thresholdrange{K}} f(i,z,t; {\bm p}).
\end{equation}
Once we obtained the values of $f(i, z, t)$, we can derive the stationary probability distribution for the age of information -- at the upload time -- for the data collected at location $i$, namely  $\Delta_i({\bm p})$,  
\begin{equation}\label{eq:CCDF}
F_{\Delta_i}(d):=\mathbb{P}( \Delta_i({\bm p}) > d)= \sum_{t=d+1}^{\thresholdrange{K}}  \sum_{z\in {\mathcal L}}   f(i, z, t; {\bm p}). 
\end{equation}
Relation \eqref{eq:CCDF} provides an important measure for SP: it is the probability that an input shadow price vector can meet the requirements for the IoT data collected at a tagged location. Furthermore starting from $F_{\Delta_i}(d)$, it is possible to evaluate the deviation of the age of collected data from its average value, e.g., by using Chebyshev inequality.

Finally, the expected age of collected data from location $i\in {\mathcal L}$  is given by
\begin{equation}\label{eq:E}
 \E [ \Delta_i({\bm p})] = \sum_{t=0}^{\thresholdrange{K}}  \sum_{z\in {\mathcal L}}  t \cdot f(i, z, t;\bm{p}) 
\end{equation}

The  \JOAC problem is a constrained non-linear integer valued optimization problem defined over the set of multi-threshold policies. Finding 
a solution is made difficult because the structure of function $Y$ is not convex over the shadow price vectors ${\bm p}$. In what follows, we propose a heuristic algorithm which utilizes  the structure of the devices' optimal strategy to solve the problem. 


\section{Algorithms for Optimal Pricing} 
\label{sec:algorithms}

We introduce   efficient algorithms to solve the joint traffic offloading and aging control pricing problem. 
The algorithms are driven by the rationale according to which a shadow price vector should permit to offload as much traffic as possible towards locations with smaller
 costs.   In order to obtain the optimal pricing, we begin
 by showing that optimal prices correspond to optimal users thresholds, allowing us to simplify analysis through the control of thresholds rather than prices (Section~\ref{sec:prices2thresholds}).
 Then, we consider a Markov Chain Monte Carlo (MCMC) approach to find  the optimal thresholds (Section~\ref{sec:mcmc}), followed
 by  its simulated annealing (SA) extension -- a standard technique for constrained combinatorial optimization problems~\cite{pincus70,connnors88} (Section~\ref{sec:sa1}).
 The special nature of our problem  allows us to further refine the SA solution leveraging the independence of nodes that are geographically far apart (Section~\ref{sec:sa2a} and \ref{sec:sa2b}).

\subsection{From prices to thresholds}
\label{sec:prices2thresholds}

In the \JOAC  problem introduced in the previous section,  shadow prices set by SP are our control variables.
Next, we argue that thresholds set by users 
can alternatively be taken as our controls. Indeed,  SP prices impact users thresholds, and users thresholds  impact load at different locations. Hence,
framing the problem exclusively based on users thresholds rather than prices simplifies the analysis.

Let $\threshold_l$ be the AoI threshold corresponding to location $l$.
A threshold   $\threshold_l=t$  means a device uploads   data collected  at location $l$ only if its age exceeds $t$.  Then, the  threshold vector    ${\thresholdvector}$ is an $L$ dimensional vector given by    ${\thresholdvector} =(\threshold_1, \threshold_2, \ldots, \threshold_L)$,  comprising one threshold per location. 

Let $\threshold_{max}$ be the maximum threshold  induced from all  pricing vectors in ${\mathbb R}^L$.  
Then,  $\threshold_{max}$ is given by 
\begin{equation}
\threshold_{max} = \max\{t\in {\mathbb N} \mid  \mathbb{P}(\Delta_i (t, {\bm 0}_{-i})>d) <\epsilon, \forall i \in {\cal L} \}
\end{equation}
where $\Delta_i (t, {\bm 0}_{-i})$ is the age of data collected  at location $i$, in a setup wherein all locations except $i$ correspond to threshold $\threshold_i=0$.
  Let $S$ be the set of feasible threshold values, i.e., $S=\{0,1,\ldots, \threshold_{max}\}$. 

Finally, the threshold-based  \JOAC is given as follows:
\begin{align}
 \textsc{Threshold-based JOAC (T-JOAC): } &  \nonumber  \\
&\hspace{-90pt} \min_{{\thresholdvector}\in {\thresholdvector}_{\epsilon,d}}  W({\thresholdvector})  :=  \sum_{i\in {\mathcal L}} Y_j({\thresholdvector})  C_j\l   \\
&\hspace{-90pt}  y_{ij}({\thresholdvector}) \geq 0,\;  \forall i, \forall j \in {\mathcal L}  
\label{unP1}
\end{align}
where 
\begin{equation}
\thresholdvector_{\epsilon,d} =\{\thresholdvector\in  S^{L} \mid \mathbb{P}(\Delta_i (\thresholdvector)>d) <\epsilon,\;  Y_{i}(\thresholdvector)  \leq B_i, \forall i  \}. \label{eq:S}
\end{equation}

In the above formulation, the objective function corresponds to~\eqref{eq:cy}-\eqref{C0} in  \JOAC. The constraints~\eqref{unP1} and~\eqref{eq:S}
capture~\eqref{C2b} and~\eqref{C0}-\eqref{C2}, respectively.

Let ${\thresholdvector}^*_{\epsilon, d}$ be the set of optimal threshold vectors,
\begin{equation}
{\thresholdvector}^*_{\epsilon, d} = \{\thresholdvector\in \thresholdvector_{\epsilon, d}  \mid  W(\thresholdvector) = \min_{\thresholdvector'\in \thresholdvector_{\epsilon, d} } W(\thresholdvector')\}.  \label{eq:sstar}
\end{equation}
Next, we present efficient algorithms to find elements in ${\thresholdvector}^*_{\epsilon, d}$.


\subsection{Markov Chain Monte Carlo (MCMC)} \label{sec:mcmc}

MCMC starts from a feasible solution and attempts to improve it by performing random perturbations. 
  A key feature of MCMC is the use of trial and error to avoid being trapped at local minima. Furthermore, it is simple to implement in a distributed way. 

Given the current state, 
the procedure generates a trial state at random and evaluates the objective function at that state. If the trial state  improves the objective function,  i.e., 
if the objective function evaluated at  the trial state is better than at the current state,
 the system jumps to this new state. Otherwise, the trial is accepted or rejected 
 based on a certain probabilistic criterion. The main feature of the procedure is that a worse off solution may be accepted as a new solution with a certain probability.

%


Next, we introduce the  Boltzmann-Gibbs distribution corresponding to    \TJOAC,  
\begin{equation}\label{eq:BG}
\pi_T({\thresholdvector}) = \frac{1}{Z} \exp^{-W({\thresholdvector})/T},
\end{equation}
where $Z$ is a normalization constant 
\begin{equation}\label{eq:BG}
Z = \sum_{{\thresholdvector} \in S^L} \exp^{-W({\thresholdvector})/T},
\end{equation}
and $T$ is a constant, referred to as the temperature, and  whose discussion  is deferred to the upcoming section.

MCMC has multiple flavors. Next, we consider the  most common MCMC method, namely Metropolis--Hastings (MH)\cite{MH}.
One of the ingredients of MH is a transition matrix $Q^*$ for \emph{any} irreducible discrete time Markov chain (DTMC). The states of the DTMC 
are given by the reachable \TJOAC threshold vectors (see~\eqref{eq:S}).    
Chain $Q^*$ is the \emph{proposal chain}, as samples collected from $Q^*$ are the proposal  threshold vectors. Then, based on those proposals,
the MH algorithm decides whether or not they will be accepted.

Although the algorithm works for any irreducible proposal chain, the choice of the chain impacts its time to convergence.  
We begin by considering the simplest   chain, whose transition matrix is uniform and symmetric.  Noting that the  threshold vector is  an $L$-dimensional vector
 $\thresholdvector$,
in the simplest setting, we allow every single component of the vector to be updated conditional on the other $L-1$ components being fixed and given.  
In this case, the algorithm is also known as Gibbs sampler, and is a special case of the MH algorithm.  
 Then, in Sections~\ref{sec:sa2a} and~\ref{sec:sa2b} we indicate how to leverage spacial information
to refine the proposal matrix and reduce convergence time by allowing multiple dimensions of vector $\thresholdvector$ to be updated simultaneously. 

 
 The proposal chain is given by  $Q^*$. Let $q^*({\thresholdvector}, {\thresholdvector}')$ be the entry at position $({\thresholdvector}, {\thresholdvector}')$ of the corresponding transition matrix:
\begin{align}
q^*({\thresholdvector}, {\thresholdvector}')= \left\{
\begin{array}{ll}
\frac{1}{L \, (\threshold_{max}-1) }, &  \exists i:  \threshold_i\not=\threshold_i' \textrm{ and } \threshold_j=\threshold_j', \forall j \neq i, \\
0, & \textrm{otherwise.} \\
\end{array}
\right.
\label{tran}
\end{align}
Clearly, $\sum_{{\thresholdvector}'}q^*({\thresholdvector}, {\thresholdvector}') = 1$ as each transition corresponds to the change of one of the $L$ thresholds to one of the distinct $\threshold_{max}-1$ values.

Let $\delta({\thresholdvector},{\thresholdvector}')$ be the   change in the objective function when going from  ${\thresholdvector}$ to ${\thresholdvector'}$,
\begin{equation}
\delta({\thresholdvector},{\thresholdvector'})=W({\thresholdvector}')-W({\thresholdvector})=\sum_{j} \delta_j({\thresholdvector},{\thresholdvector'}) \label{eq:deltasum}
\end{equation}
where
\begin{equation}
\delta_j({\thresholdvector},{\thresholdvector'})= ( Y_j({\thresholdvector}') - Y_j({\thresholdvector}))C_j. \label{eq:deltaj}
\end{equation}

To describe the Markov chain ${\thresholdvector}{(0)}, {\thresholdvector}{(1)}, \ldots$,  assume that the threshold vector at iteration $t$ is given by  ${\thresholdvector}$. Then,  the threshold vector    is determined as follows
\begin{enumerate} 
\item choose     threshold vector   ${\thresholdvector}'$ according to $Q^*$, i.e., choose       ${\thresholdvector}'$ with probability given by~\eqref{tran}.  
Threshold vector   ${\thresholdvector}'$ is the \emph{proposal  threshold vector};
\item let the \emph{acceptance function} be given as follows,
\begin{equation}
\tilde{a}({\thresholdvector},{\thresholdvector}')=\frac{\pi_T({\thresholdvector}')}{\pi_T({\thresholdvector})}=e^{-\delta({\thresholdvector},{\thresholdvector}')/T}. \label{eq:acceptancef}
\end{equation}
If $\tilde{a}({\thresholdvector},{\thresholdvector'}) \ge 1$, i.e., if $\delta({\thresholdvector},{\thresholdvector}') \le 0$, then $\thresholdvector'$ is accepted, ${\thresholdvector}{(t+1)} \leftarrow {\thresholdvector'}$. Otherwise, it is accepted with probability $\tilde{a}({\thresholdvector},{\thresholdvector'})$ and rejected otherwise. If rejected, the threshold vector remains unchanged,  ${\thresholdvector}{(t+1)} \leftarrow {\thresholdvector}(t)$. 
\end{enumerate}

The proof of the following statement is reported in Appendix~\ref{app:lemma1}.
\begin{lemma}
\label{lem1}
The Markov chain produced by the above algorithm  has stationary distribution $\pi_T$. 
\end{lemma}
In what follows, we extend the above algorithm through a simulated annealing approach.

\subsection{Simulated annealing}
\label{sec:sa1}

Next, we consider simulated annealing. It corresponds to allowing $T$ to decrease in time, in order to guarantee the convergence to an optimal threshold vector (and corresponding pricing). 
Let $T_t$ be the temperature at iteration $t$.
 For a temperature  $T_t$, we define an inhomogeneous Markov chain $(Y_t)_{t\in {\mathcal N}} $ with transition kernel $Q_{T_t}$ at time $t$. 
 If $T_t$ decays to zero sufficiently slowly, the Markov chain $Q_{T_t}$ will reach a sufficiently small neighborhood of the target equilibrium,  $\pi_{{T_t}}$. 
 For this reason $T_t$ is called the cooling schedule of SA. In this paper we use the standard cooling schedule in the form  $T_t=\frac{\hat{a}}{\log(1+t)}$ where $\hat{a}>0$ is a constant that determines the cooling rate order.
\begin{theorem}
\label{theom_conv}
If $T_t$ assumes the parametric form 
\begin{equation}
T_t =\frac{\hat{a}}{\log(t+1)}
\end{equation}
where  
\begin{equation}
\hat{a}= \frac{N F}{\Tseconds} \max_{j\in{\cal L}} C_j,
\end{equation}
 then 
\begin{equation}
\lim_{t\to \infty} \pi_{T_t}(\{{\thresholdvector}\in {\thresholdvector}^*_{\epsilon, d}\})=1
\label{conv}
\end{equation}
\end{theorem}
The above theorem shows that the Markov chain with  transition  matrix $Q_{T_t}$ converges to an   optimal threshold $s^*\in S^*$,
 where $S^*$ is the set of the optimal solutions of the T-\JOAC problem (see~\eqref{eq:sstar}).
\begin{proof}
Our proof   is based on the technique introduced in \cite{hajek88}.  
Note that the objective function $W$ is nonnegative and its maximum   value   is attained when data is collected by all devices at locations where  cost is maximal.
 Thus $\hat{a} >  W({\thresholdvector})$ for all states ${\thresholdvector}$. In particular, 
 letting $d^*$ denote the maximum value of $ W({\thresholdvector})$ at all states ${\thresholdvector}$ which correspond to a local
 but not global minima,  
 we have $\hat{a} > d^*$.  Then,
\begin{eqnarray}\label{eq:suffcond}
\sum_{t=1}^{\infty} \exp^{-d^*/T_t}&=& \sum_{t=1}^{\infty} \exp^{(-\frac{d^*}{\hat{a}} \log(t+1))}\\
& >&   \sum_{t=1}^{\infty} \frac{1}{t+1} = +\infty  
\end{eqnarray}
Theorem 1 in \cite{hajek88} ensures that under the above condition the limit   \eqref{conv} holds, which completes the proof. 
\end{proof}

\begin{algorithm} \small
	\DontPrintSemicolon
	\SetKwComment{comment}{$\triangleright$ \ }{}
	\KwIn{$\mathcal{L},S,{\thresholdvector}(t-1), i$ \comment*[r]{$i$ is the candidate location for AoI threshold change}}
	\SetKwBlock{init}{Initialization}{end}
	\SetKwBlock{assignment}{Assignment Phase}{end}
	\SetKwBlock{test}{Test Phase}{end}
	\SetKwBlock{decision}{Decision Phase}{end}
	\SetKwBlock{stop}{Stop Condition}{end}
	\SetKwBlock{finout}{Final output}{end}
	\assignment{
		Set temperature $T_t={\hat{a}}/{\log(1+t)}$\;
	Set    new threshold $\threshold'_i $ uniformly at random,  $\threshold'_i \in  S\textbackslash\{\threshold_i(t-1)\}$\;
	Set  $\threshold'_j  \leftarrow \threshold_{j}(t-1)$,  $\forall j\in {\mathcal L}\textbackslash\{i\}$\;
	}
	\test{
	At each $j \in \mathcal{N}_i$ locally measure $\delta_j$ (see~\eqref{eq:deltaj}) and $\mathbb{P}(\Delta_j({ \thresholdvector}')>d)$\; 
	Each  location  $j \in \mathcal{N}_i$ sends measurements to location $i$\;
	}
	\decision{
	\If{  ${\thresholdvector}'$   does not satisfy   constraints for all locations in $ \mathcal{N}_i$}{go back to selection phase (line 4)}
	Set  $\delta = \sum_{j\in \mathcal{N}_i \cup \{i\}} \delta_{j}(\bm{\tau}(t-1),\bm{\tau}')$ (see~\eqref{eq:deltasum})\;
	$\threshold_i(t) \leftarrow \threshold_i(t-1)$\;
	\eIf{$\delta \leq 0$}{$\threshold_i(t) \leftarrow \threshold'_i$}{$\threshold_i(t) \leftarrow \threshold'_i$   with probability  $e^{-\delta/T_t}$ (see \eqref{eq:acceptancef})}
	}
	\KwOut{$\threshold_i(t)$}
	\caption{Simulated Annealing (SA) algorithm for \TJOAC at time $t$, at the neighborhood of location $i$}
	\label{algo}
\end{algorithm}

\subsection{Simulated annealing leveraging neighborhoods} \label{sec:sa2a}

\subsubsection{Neighborhood structure}
Next, we leverage the neighborhood structure between locations to specialize SA to our T-JOAC problem. Let the neighborhood set ${\mathcal N}_i$ for  location $i$ be defined as follows: a  location $j$  belongs  to ${\mathcal N}_i$ if $j$ is located within a given radius such that data offloaded to $j$ can be impacted by traffic generated at location $i$.

Note that the neighborhood structure 
 depends primarily on the geographic position of the locations, mobility of the devices, and the maximum time  that a device can  wait before uploading the data.
 Indeed, let $\thresholdvector_{max}$ be a threshold vector wherein all elements equal $\threshold_{max}$.  Such threshold vector corresponds to nodes that defer transmissions as much as possible.  Then, the neighborhood of location $i$ is  defined as follows,
 \begin{equation*}
\mathcal{N}_i =\{ j \in \mathcal{L} \ | \ Y_{ij}(\thresholdvector_{max}) > 0 \ \textrm{ or } \  Y_{ji}(\thresholdvector_{max}) > 0   \}.
\end{equation*}
Indeed, if   traffic at locations $i$ and $j$ does not interfere with each other under the extreme scenario where all thresholds are set to their maximum values, one can safely assume that locations $i$ and $j$ are not neighbors.  
%
%
%
%

 Let  $G(V, E)$ be the location neighborhood graph, where  $V$ is the set of vertices representing the locations and $E$ is the set of  edges, where an edge is a link between two vertices  indicating that  the two corresponding locations are neighbors.

\subsubsection{Simulated annealing leverage neighborhoods}
Hereafter, we describe the detailed implementation of the  simulated annealing algorithm for solving the  T-\JOAC  problem (see Algorithm~\ref{algo}).   
Time is divided into discrete slots. At the first slot, we begin by initializing the thresholds of all locations to zero, which corresponds to a price $p=0$.  Then, at each   time slot $t$,  we let  $T_t=\hat{a}/{\log(1+t)}$ (the initial temperature $T_t$ should be large enough to allow all candidate solutions to be accepted uniformly at random), and the system goes through three phases: assignment, testing, and decision.  
The SP selects a location  $i \in \mathcal{L}$ uniformly at random and run Algorithm~\ref{algo}.  During the assignment phase, the threshold of location $i$ is modified, while  letting all other thresholds unchanged (lines $1-4$).  At the test phase, the SP receives measurements from all locations which are possibly affected by a change in the threshold of location $i$, i.e., from all $j \in \mathcal{N}_i$, and checks whether the newly generated threshold vector ${\thresholdvector}'$ is feasible (lines $5-10$). If it isn't feasible, the algorithm returns to the selection phase. Otherwise, it continues in the decision phase, by assessing the    change in the objective function, $\delta$, again using data from $j \in \mathcal{N}_i \cup \{i\}$  (line $11$).  If the change is negative,  the new threshold vector is accepted (lines $13-14$). Otherwise, it is accepted with probability $\exp(-\delta/T_t)$.   SP repeats this procedure until the established stopping conditions  are satisfied. i.e., either threshold vector is not changed for two successive time slots or $T_t < \varepsilon$. 

\subsubsection{Independent sets}
To accelerate the SA algorithm, we  exploit  independent sets of locations, i.e., a partition of locations into sets where locations within each set are not affected by a change of threshold  that may occur in other locations of the same set. 
In the following paragraph,  we indicate  how the proposal chain can be adapted to account for independent sets of locations, under a serial implementation.  
In practice, the speedup is obtained since the algorithm can be run in parallel for all the locations that belong to the same independent set, as indicated in Section~\ref{sec:sa2b}.   Our experiments demonstrate that this parallelization can attain a two-fold speedup of the run time with respect to the basic implementation of the algorithm.

\subsubsection{Proposal chain leveraging the neighborhood structure} \label{sec:newproposal}
Given the neighborhood structure, we adapt the proposal chain introduced in~\eqref{tran} in order to allow for multiple threshold adjustments at the same iteration. The new proposal chain $\tilde{Q}^*$, whose  $({\thresholdvector}, {\thresholdvector}')$ entry is denoted by  $\tilde{q}^*({\thresholdvector}, {\thresholdvector}')$, is given as follows: 
\begin{align}
\tilde{q}^*({\thresholdvector}, {\thresholdvector}')= \left\{
\begin{array}{ll}
\frac{1}{|{\mathcal M}_{\thresholdvector}| }, &   \textrm{if } {\thresholdvector} \sim {\thresholdvector}' \\
0, & \textrm{otherwise} \\
\end{array}
\right.
\label{tran1}
\end{align}
where
\begin{equation}
\mathcal{M}_{\thresholdvector}=\{ {\thresholdvector}' | {\thresholdvector} \sim {\thresholdvector}'  \}
\end{equation}
and ${\thresholdvector} \sim {\thresholdvector}' $ denotes that threshold vectors ${\thresholdvector}$ and ${\thresholdvector}'$  are adjacent.  Two threshold vectors are adjacent if they differ in at least one position and, in addition, all positions that differ across the two threshold vectors correspond to locations that belong to the same independent set, i.e., in the location neighborhood graph $G(V,E)$ there is no edge between the locations whose thresholds differ.
If each location corresponds to its own independent set, i.e., if we have $L$ independent sets, then $|\mathcal{M}_{\thresholdvector}|=L (\threshold_{max}-1)$ and the above proposal chain reduces back to~\eqref{tran}.

Note that the above proposal chain produces proposal threshold vectors wherein multiple thresholds may change concomitantly with respect to the current threshold vector.  Then, a straightforward adaptation of Algorithm~\ref{algo} accepts or rejects the proposal threshold vector as a whole, treated as a single entity.
  
  In the following section, in contrast, we treat each of the neighborhoods independently.  In particular, at each step of the algorithm, we have multiple new proposal thresholds, which are evaluated in parallel and may be independently accepted
or rejected. 
 Even if one neighborhood rejects a particular proposal for a new threshold, other independent neighborhoods may accept their proposals.

\subsection{Accelerated simulated annealing with parallel computations: a  coloring approach} \label{sec:sa2b}

\begin{table}[t]\caption{Table of notation for coloring algorithm and \TJOAC}
\centering
\begin{tabular}{p{0.3\columnwidth}|p{0.6\columnwidth}}
\hline
{\it Variable} & {\it Description}\\
\hline
$H = \{h_1, \ldots, h_L \}$& set of colors that can be assigned to a location\\
$c_l \in H$ & color of location $l$ \\
$\phi(\bm{c})$& set of colors used by $\bm{c}$\\
${\bm{c}}^{^{(n)}}$&  current coloring vector  in Algorithm~\ref{algo2}  \\
${\bm{c}}^\star$&  best coloring vector so far (broadcast from Algorithm~\ref{algo2} to Algorithm~\ref{ASAA})  \\
${\bm{c}}(t)  $& current coloring vector in Algorithm~\ref{ASAA}  \\
\hline
\end{tabular}\label{tab2:notation}
\end{table}

\begin{algorithm} \small
	\DontPrintSemicolon
	\SetKwBlock{init}{Initialization}{end}
	\KwIn{$\mathcal{L}, H, b,  {\bm{c}}^{(0)}$}
$n \leftarrow 1$

$\bm{c}^\star \leftarrow {\bm{c}}^{(0)}$

\While{true}{
	Set  $\tilde{T}_n = b/\log(1+n)$\;
	Select location $l\in {\cal L}$ and  color  ${c}'_l \in H\textbackslash\{c_l^{(n-1)}\}$\; 
	Set ${c}'_j \leftarrow c_j^{(n-1)} \quad \forall j \in {\mathcal L} \textbackslash \{l\}$\;
	\If{$\bm{c}'$ is not feasible}{go to step 5}
	Set $\beta  = |\phi(\bm{c}')|-|\phi(\bm{c})|$\;
	$\bm{c}^{(n)}\leftarrow \bm{c}^{(n-1)}$\;
	\eIf{$\beta \leq 0$}{$\bm{c}^{(n)}\leftarrow \bm{c}'$}{$\bm{c}^{(n)}\leftarrow \bm{c}'$ with probability  $e^{-\beta/\tilde T_n}$}
	\If{$|\phi(\bm{c}^{(n)})| <|\phi(\bm{c}^\star)|$  }{$\bm{c}^\star \leftarrow \bm{c}^{(n)}$\;
	broadcast   new coloring   $\bm{c}^\star$ to all locations\; }
	$n \leftarrow n+1$
	}  
%
		\caption{SA-based coloring agorithm}
	\label{algo2}
\end{algorithm}


\subsubsection{Colorings} \label{sec:colorings}
Fewer independent sets correspond to more opportunities for concomitant threshold adjustments. 
A partition of the locations into independent sets can be obtained as the result of a graph coloring procedure. The coloring of a graph is a function that assigns different colors to adjacent vertices of a graph. 

 Given  a graph $G$, the Graph Coloring Problem (GCP) seeks the minimum number of colors $\chi(G)$ which can be used to color $G$. Such a number is called the chromatic number. An upper bound for the chromatic number is given by 
 the maximum vertex degree plus one, and is attained by a  greedy coloring procedure. Nonetheless,  such upper bound (may be loose) and optimal coloring is an $\mathcal {NP}$-hard problem, motivating heuristic solutions.  We observed that simulated annealing can be used as a heuristic to solve GCP~\cite{SAcoloring2} to obtain a near optimal solution.  We can thus account for the coloring process and the resulting parallelization to \TJOAC in an extension to Algorithm~\ref{algo}.

\subsubsection{SA for colorings}

The location coloring is obtained by using a specialized simulated annealing algorithm. Let $H = \{ h_1,h_2,...,h_L \}$ be the set of available colors. 
 Then, the 
 %
coloring (vector) $\bm c$ indicates, for each location, its corresponding color, i.e., $c_l=h_m$ if the color of location $l$ equals $h_m$. 

 A coloring is {\em feasible} if, for any pair of locations $i, j$ such that $i\in {\mathcal N}_j$, we have $c_j\not = c_i$.  
  Let $\phi({\bm c})$ be the set of  colors used by coloring  ${\bm c}=(c_1,c_2,..,c_L)$. Thus the objective of the  coloring problem  is to find a feasible solution that minimizes  the cardinality of set $\phi({\bm c})$. The algorithm begins with an initial feasible coloring scheme, then improve it further using simulated annealing. The initial coloring schemes could be achieved with a random allocation of one color to each location, where the number of colors are equal to the number of locations or an intermediate solution is obtained by greedy algorithm. The temperature  $\tilde{T}_n$ is structurally similar to  $T_t$ used in Algorithm~\ref{algo}, noting that now  $b$ plays the role of $\hat{a}$, and corresponds  to an upper bound on the  number of colors to be adopted. In the simplest setting, we let $b=L$.

Algorithm~\ref{algo2} is used to continuously search for better colorings.
 In lines $5$ and $6$ 
  the algorithm chooses a location $l$ uniformly at random, and  a color from $H\textbackslash\{c_l^{(n-1)}\}$. Lines $7$ and $8$ produce the proposal coloring vector ${\bm c}'$. 
   Lines $9-15$ test if the proposal is accepted or not. Finally, if the new coloring vector uses fewer colors than the current best candidate, the new coloring vector is broadcasted to all locations.  

 \begin{algorithm}[h!] \small
    \DontPrintSemicolon 
    \SetKwBlock{DoParallel}{run in parallel}{end}
    \SetKwBlock{init}{Initialization}{end}
    \SetKwBlock{SA}{Optimal threshold search}{end}
    \SetKwBlock{coloring}{location coloring}{end}
    \KwIn{$\mathcal{L}, H, S,  \varepsilon,  {\thresholdvector}(t-1), {\bm c}(t-1),  {\thresholdvector}^\star$}

          ${\bm c}(t) \leftarrow {\bm c}(t-1);$ \qquad ${\thresholdvector}(t) \leftarrow {\thresholdvector}(t-1)$

    \If{there exists untreated broadcasted ${\bm c}^\star$}{    
     update colorings at all locations\;
      ${\bm c}(t) \leftarrow {\bm c}^\star$} 
     
     Choose uniformly at random a color $h(t)  \in \phi({\bm c})$\;
    
%
    \SA{
    \DoParallel{
    
    run Algorithm~1  at neighborhood of first location with color $h(t)$\;
    
    
    $\quad \vdots$

        run Algorithm~1  at neighborhood of last location with color $h(t)$\;

    }
    
    Update               ${\thresholdvector}(t)$ given outputs from parallel runs of Algorithm~1

	}

\If{$W({\thresholdvector}^\star) > W({\thresholdvector}(t))$}{          ${\thresholdvector}^\star \leftarrow {\thresholdvector}(t)$} 	
	
\If{threshold vector is not changed for two successive time slots or $T_t < \varepsilon$}{stop}

    \KwOut{$ {\thresholdvector}(t)$, $ {\thresholdvector}^\star$ and ${\bm c}(t)$}
    \caption{Accelerated SA algorithm (at time slot $t$)}
    \label{ASAA}
\end{algorithm}   
   
   In summary, Algorithm~\ref{algo2} is continuously run by the SP, e.g.,  at a fast time scale, and as improvements are found they   are  broadcasted to the locations which locally run  Algorithm~\ref{algo}.   
   In particular, the time scale at which   Algorithm~\ref{algo2} is executed, whose iterations are denoted by $n$,  is decoupled from the  scale of the time slots considered in Algorithm~\ref{algo}, denoted by $t$.   
   The integrated solution involving  Algorithms~\ref{algo} and~\ref{algo2} is presented in Algorithm~\ref{ASAA}, and is described in the sequel.

\subsubsection{SA for T-\JOAC with colorings}

\label{sec:finalalgo}

The accelerated simulated annealing for \TJOAC is obtained by leveraging the colorings produced by Algorithm~\ref{algo2}, integrating them with Algorithm~\ref{algo} as shown in Algorithm~\ref{ASAA}.  As explained previously, the SP continuously broadcasts colorings ${\bm c}^\star$ obtained using Algorithm~\ref{algo2}. At line $5$ of Algorithm~\ref{ASAA}, the SP selects a color $h(t)$ from set of colors $\phi(c)$ currently in use. Then,  Algorithm~\ref{algo} is run locally (in parallel) at locations colored with the same color $h(t)$ (lines $7-10$ in Algorithm~\ref{ASAA}). As in the basic implementation, such locations generate new threshold values which are then combined into a new threshold vector ${\thresholdvector}(t)$ (line $11$  in Algorithm~\ref{ASAA}).  
By locally searching for  optimal thresholds at multiple locations 
 we improve the running time for computing  optimal thresholds, as further evidenced in the evaluation that follows.

 \section{evaluation}\label{sec6}

\subsection{Experimental setup}
 
To validate our theoretical results, we use  vehicular mobility traces of the city of Cologne (Germany), covering a region of 400 km$^2$ in a period of $24$ hours in a typical work day, involving more than $700,000$ individual vehicles. The dataset is available at \cite{cologne}. It comprises the list of users' position records, each record including a sampling timestamp, the user ID, and her position in $(x,y)$ cartesian coordinates. Positions are sampled each second. User mobility is spanned across $230$ (macro) cells and the coverage of each cell is determined a posteriori according to the Voronoi tessellation  shown in Fig.~\ref{fig:voronoi}.  In order to generate the transition probability across cells, we have restricted it to a subset of records corresponding to one hour of trace data. Then, we have resampled our dataset at $2$ second intervals to discretize the process: within such a time step the probability for a user to cross two cells is  bounded below by $0.05$.

\subsection{Aging control analysis}

In the first set of experiments, we have validated our aging control policy on real-world traces. The computation of the optimal policy using the proposed model requires estimating the transitions of the Markov chain $\Gamma$, $U(x)$ and \shadow price $p_l$ of each location $l\in {\mathcal L}$. In our reference setting, the utility of the message decays linearly over time, and remains zero afterward, $U(x)=\max(M-x,0)$. We assume that the device will collect data as soon as the existing data is uploaded. Performance of the policy is evaluated over $67$ epochs corresponding to a total duration of $134$s for a device: at each epoch, a device either uploads or defers based on the multi-threshold  policy obtained from the model.  All $230$ locations are divided into $3$ effective groups, namely range $1$, range $2$, and range $3$, corresponding each to a specific cost $C_1$, $C_2$, and $C_3$, respectively, sorted in  ascending order.  The grouping of locations is based on the congestion level, corresponding to low congestion, medium congestion, and high congestion respectively. The intuition is that the location with a higher cost should be configured with higher \shadow price according to the previous analysis. The results displayed in this section are configured with three prices $(P_1,P_2,P_3)$, $\thresholdrange{1}=0$,  two effective threshold values $(\thresholdrange{2}, \thresholdrange{3})$ and $M$ = 10.  Locations that belong to range $1$, $2$ and $3$ are assigned $P_1=0$, $P_2$ and  $P_3$, respectively.

\begin{figure}[t]
\centering
\includegraphics[scale=2]{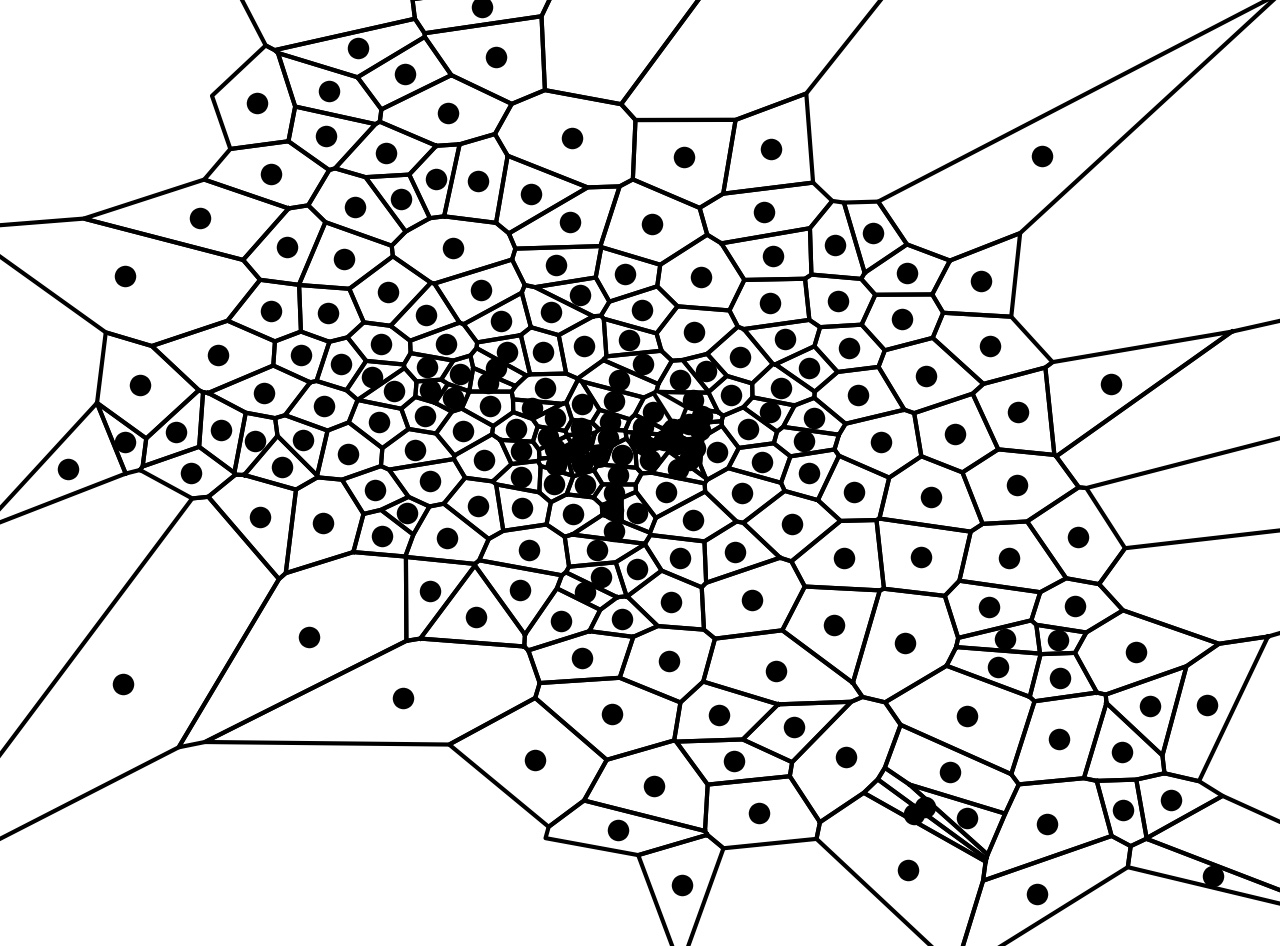}
\caption{Voronoi tessellation for the Cologne mobility trace } 
\label{fig:voronoi}
\end{figure}

Note that we considered the initial price to be zero throughout our model description and numerical experiments.  In addition, in our experiments we set the  maximum price $P_3=9$, and the intermediate price  $P_2 = 6$.  In particular, the values of $P_2$ and $P_3$ are  chosen according to our experimental goals, namely, to illustrate a threshold policy which is not degenerate to either always transmit or never transmit -- we set $P_2$ and $P_3$  with enough separation     to illustrate their  impact, while remaining  in the same order of magnitude.

Since the transition matrix $\Gamma$ is derived from traces, it is interesting to compare how the optimal policy obtained by our model compares against its alternatives. In particular, note that the mobility in the traces is neither stationary nor memoryless. Therefore, one of our goals is to assess to what extent our results still hold if those assumptions are removed. 

In Fig. \ref{fig:price}(a) we compared 1) the theoretical optimal reward predicted by our model -- i.e., using the empirical transition matrix $\Gamma$, 2) the average reward obtained simulating data collection and upload using the original real traces under the optimal policy predicted by our model and, finally, 3) the optimal reward obtained by using the multi-threshold policy calculated by exhaustive search on the real traces.  The match appears quite tight and also rather insensitive to the variation of $M$.

\begin{figure*}[t]
\centering
\includegraphics[scale=0.35]{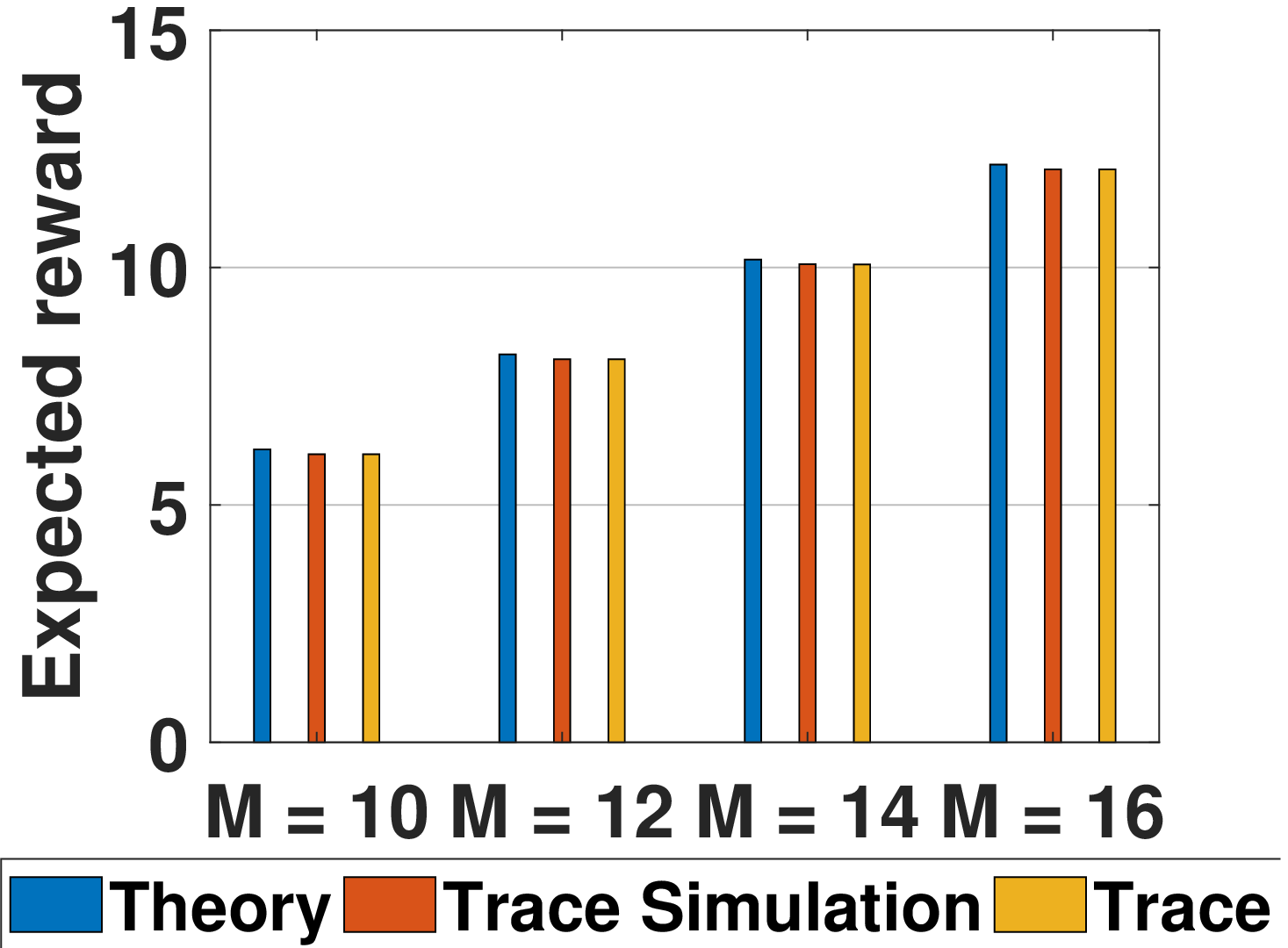} 
\includegraphics[scale=0.35]{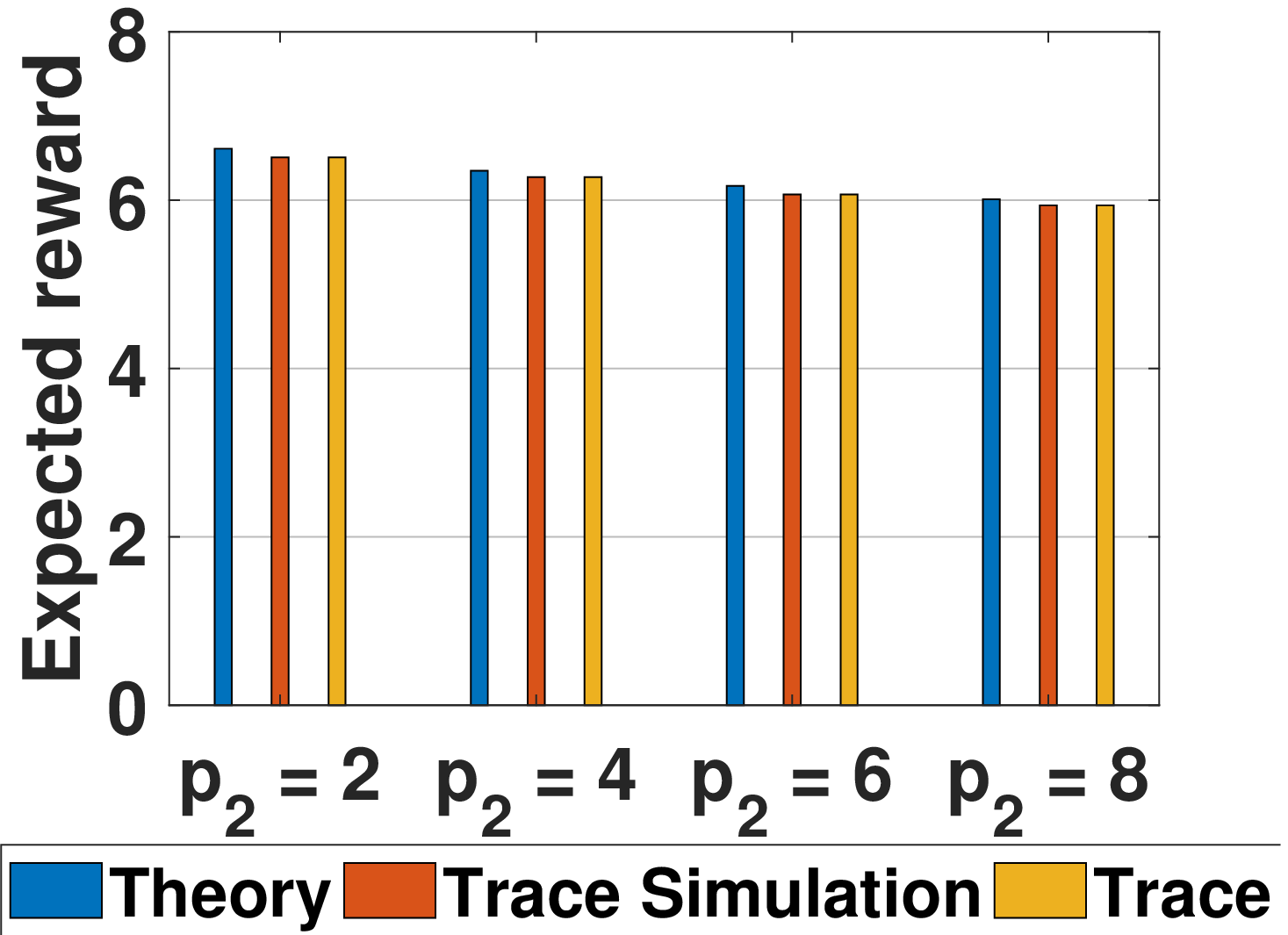} 
\includegraphics[scale=0.35]{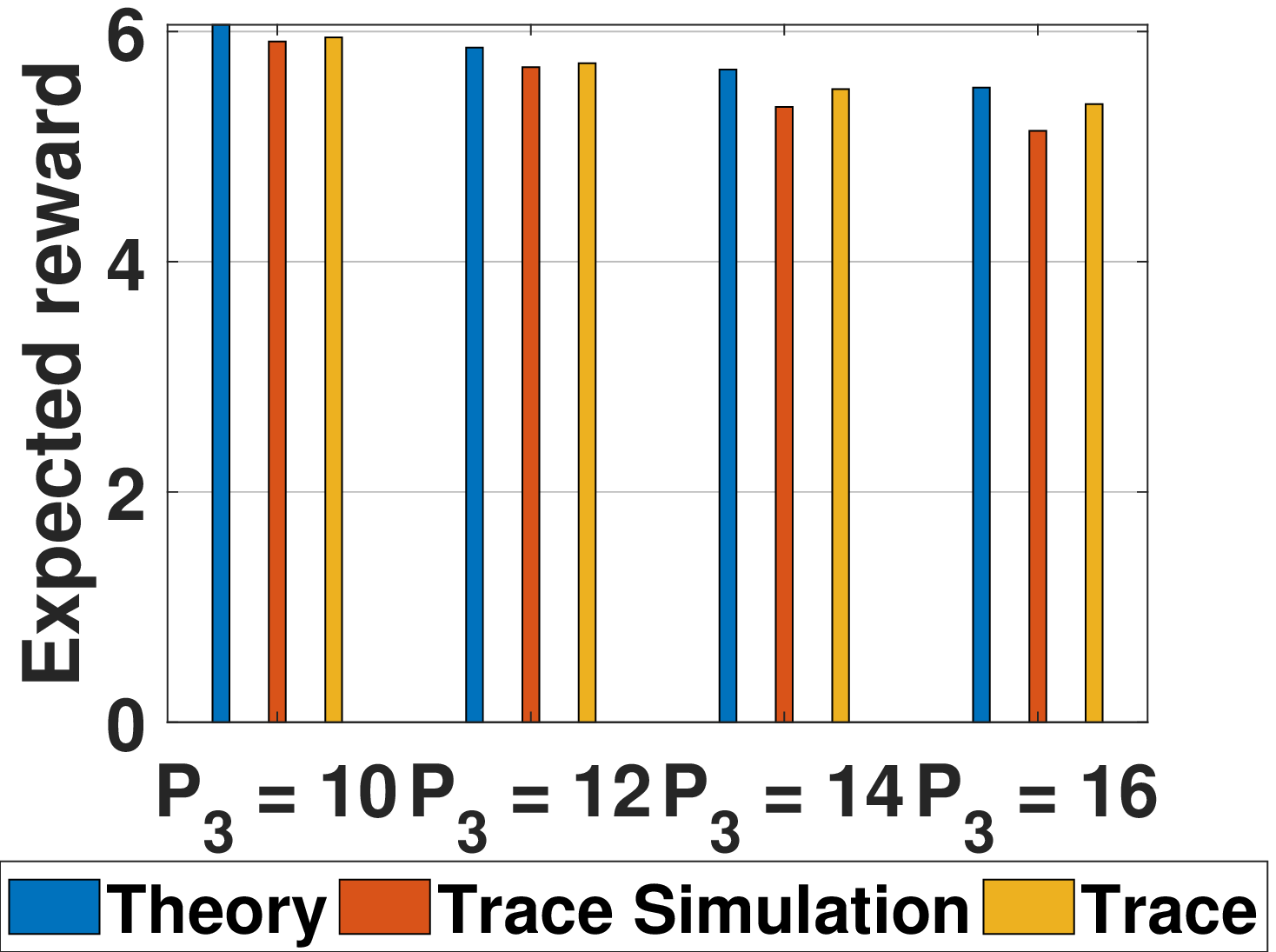}
\caption{a) Expected average rewards for i. theoretical model, ii. simulation and iii. exhaustive search; $P_1 = 0$, $P_2 = 6$ and $P_3=9$. b) Effect of \shadow price $P_2$ on the average reward; $P_1 =0$ and $P_3 = 9$; $M=10$. c) Effect of \shadow price $P_3$ on the average reward  $P_1 =0$ and $P_2 = 6$; $M=10$. }\label{fig:price}
\end{figure*}

\subsection{Offloading under unconstrained aging control}

We now explore how \shadow prices can be used to affect the freshness of information delivered by each device. To this aim, Fig.~\ref{fig:price}(b) and \ref{fig:price}(c) illustrate how the average reward changes as function of the \shadow price. Fig.~\ref{fig:price}(b) is obtained by fixing $P_1$ and $P_3$  and varying $P_2$ from $2$ to $8$. It shows that the average reward decreases with price $P_2$.  Indeed, as $P_2$ becomes larger, the device has more incentive to upload the collected data to locations in range $1$. It is possible that the age of collected data becomes higher -- i.e. upload occurs farther from the origin site -- which explains why the average reward decreases  with $P_2$. We observe the same behavior by changing the price $P_3$ and making $P_1$ and $P_2$ fixed, as depicted in Fig.~\ref{fig:price}(c). In summary, if we ignore QoS constraints, the SP should indeed increase the \shadow price for locations in range $2$ and range $3$: the effect is to shift all collected traffic to locations in range $1$, which are less expensive to lease. 

\subsection{Offloading under aging control with AoI constraints}

Fig.~\ref{fig:load}(a) shows the relative volume of traffic uploaded in range $1$, $2$, and $3$ under optimal pricing for increasing values of the AoI constraint $d$. From these simulation results, we can get some useful insights on the load balancing operated by the proposed \shadow pricing scheme. The first point on the x-axis corresponds to the profile of traffic obtained under uniform flat \shadow price, that is No Strategy (NS), meaning that all IoT traffic is uploaded where it is produced. The remaining points correspond to the optimal \shadow prices for different values of $d$.  First, we observe that it is possible to shift an important  part of traffic ($13$\%-$32$\%) towards locations in range $1$ (which corresponds to price $P_1=0$). A smaller part of traffic  ($1$\%-$12$\%) is instead shifted to locations in range $3$. When $d$ increases, we observe that more traffic is shifted since devices have more opportunities to upload collected data in locations in range $1$.

Fig.~\ref{fig:load}(b) shows how the \shadow price impacts the cost incurred by the SP. A larger value of $d$ for data generated at a given location means lower sensitivity to delay, which in turn increases the probability  to upload at locations in range $1$, which explains why the total cost is ultimately decreasing with the QoS constraint $d$ (total cost reduced by approximately 50\% when $d=18$).


\begin{figure}[h]
\centering
\includegraphics[scale=0.26]{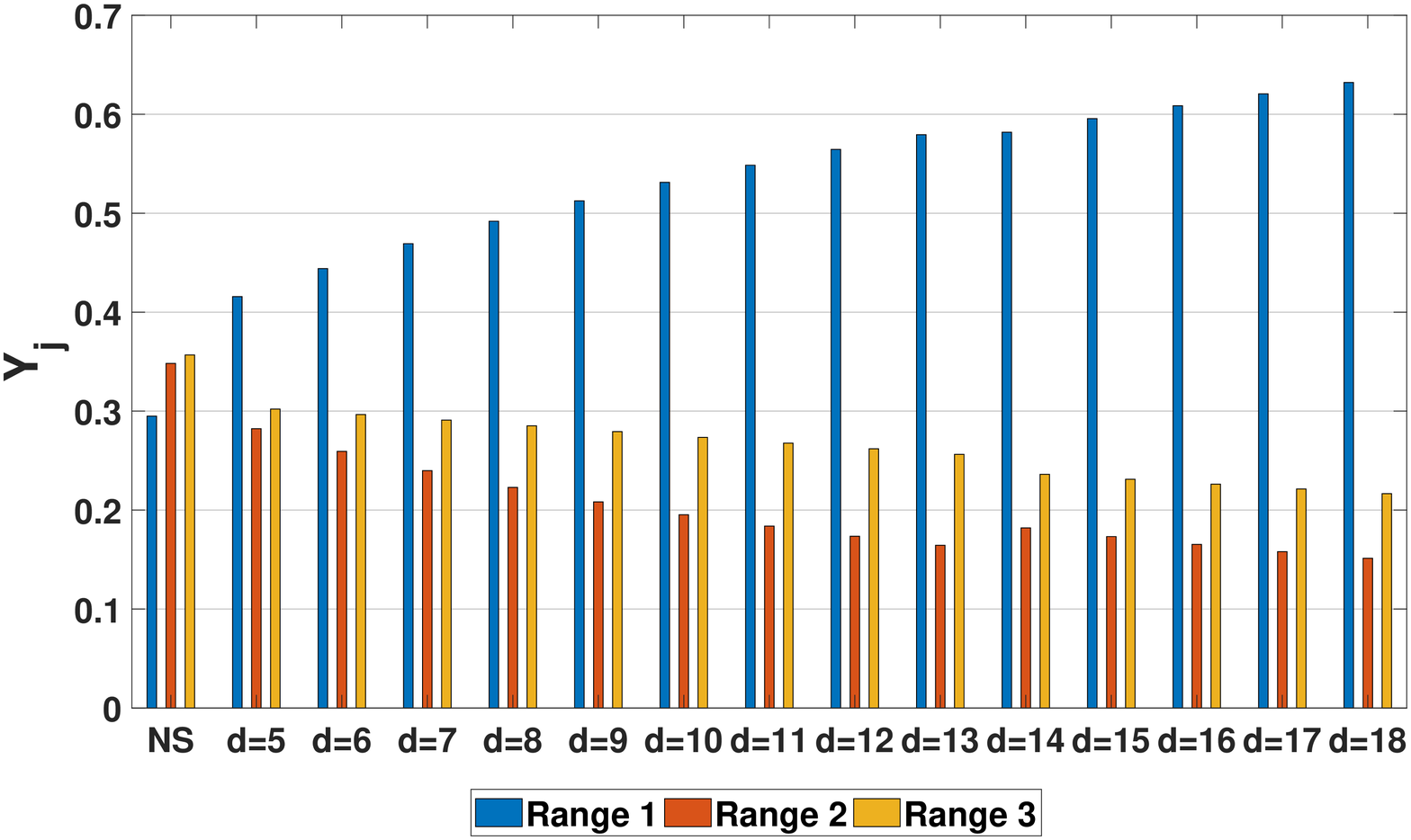}
\includegraphics[scale=0.26]{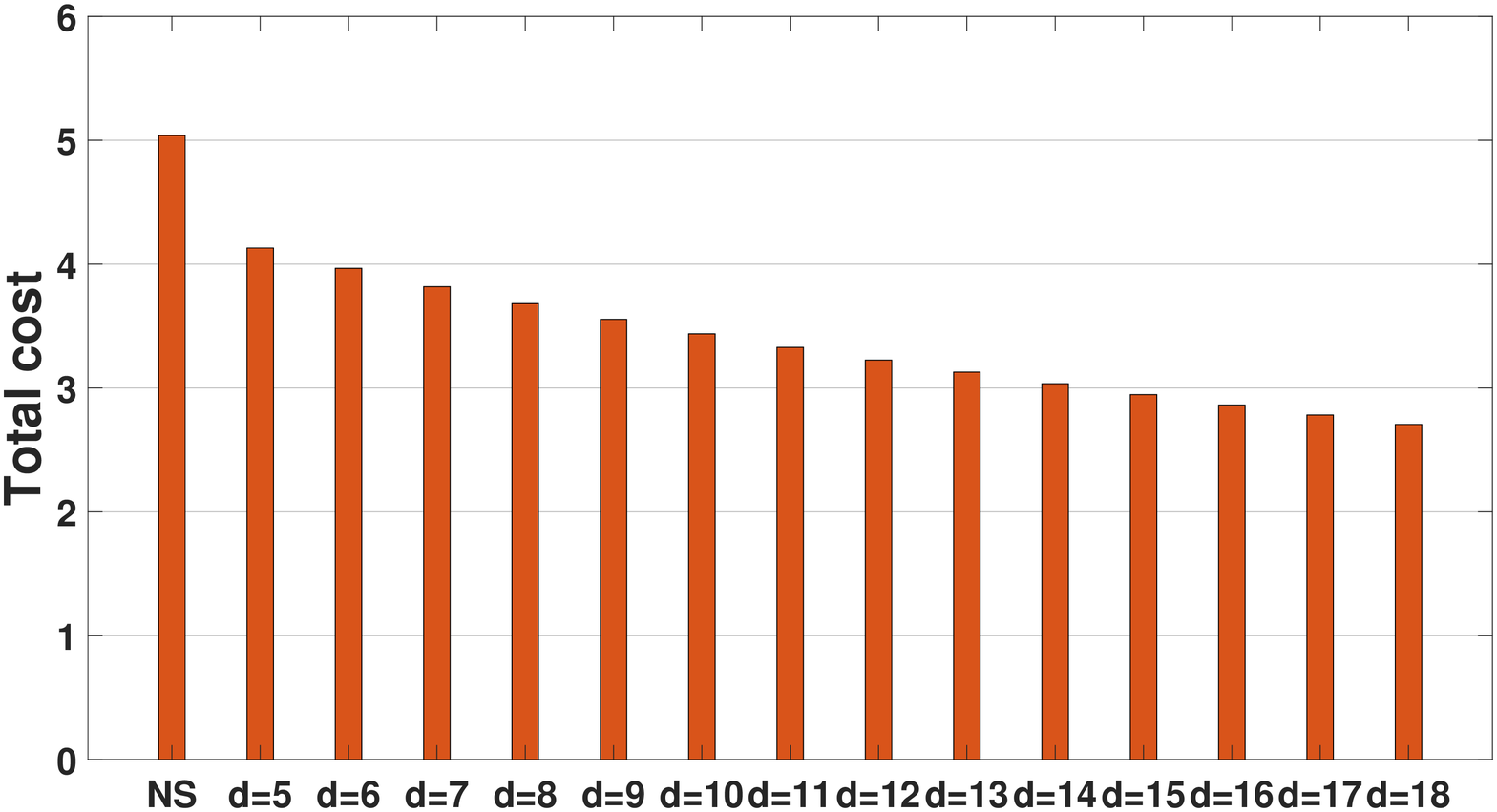}
\caption{Load balancing: a)  effect of \shadow pricing for various values of $d$; $\epsilon = 0.01$ b) cost incurred by the MSP for various values of $d$.} \label{fig:load}
\end{figure}

\subsection{Leveraging coloring for joint offloading and aging control}

Results concerning location coloring are depicted in Fig.~\ref{coloring}, where, we considered $d = 7$, $\epsilon$ = 0.01 (factor which determines the neighborhood of locations in the SA algorithm) and used the same cell topology as mentioned for earlier numerical results. Fig.~\ref{coloring} shows the convergence of the minimum number of colors for a region of $20$ locations and $230$ locations. Different temperatures are considered to test the convergence of the algorithm.  The results obtained using Algorithm 2 are compared to the integer linear programming ILP (branch-and-bound)  \cite{HANSEN} which is more conventional and can be considered as a baseline solution for coloring algorithms. As shown in Fig.~\ref{coloring}(a) SA algorithm has optimal results for the region with less locations. However, when tested on the region with full 230 locations, the SA algorithm lags behind the ILP as depicted in Fig.~\ref{coloring}(b). For the purposes of traffic offloading, the near-optimal solution to the coloring problem provided   by the SA algorithm is  acceptable as shown in the next set of experiments.

For the results shown in Fig. \ref{fig:convergence}, we used a $20$ macro cell topology; with no loss of generality, we have normalized the traffic generation as $\frac{NF}{T} = 1$.  We consider $d= 7$, $\epsilon = 0.01$, and the maximum threshold $t_{max} = d + 3$. The temperature used to control the number of iteration was derived by trying different configurations; the basic trade-off is to perform a number of iterations large enough for the algorithm to converge to the minimum value and yet bound it to a maximum value for the sake of computation time. We discovered by numerical exploration that use of $T_t = \hat{a}/\log(1 + t)$ or $T_t = \hat{a}/t^{2.8}$, where $\hat{a} = 10^6$, produces the same optimum value where the latter has faster convergence. Hence, we used the temperature setting that is best suited for the algorithm. Fig.~\ref{fig:convergence} shows the difference in the convergence of average cost using Algo.~\ref{algo} (without coloring) and Algo.~\ref{ASAA} (with coloring). As it can be observed, the proposed coloring method significantly speeds up the convergence compared to the case where the coloring method is not deployed. The gain in performance holds for all the three considered values of $d$ (Fig.~\ref{fig:convergence}:1-3). The improvement in  convergence time with respect to the SA without coloring is about 50\% , namely two fold improvement.

\begin{figure}[t]
\centering
\includegraphics[scale=0.28]{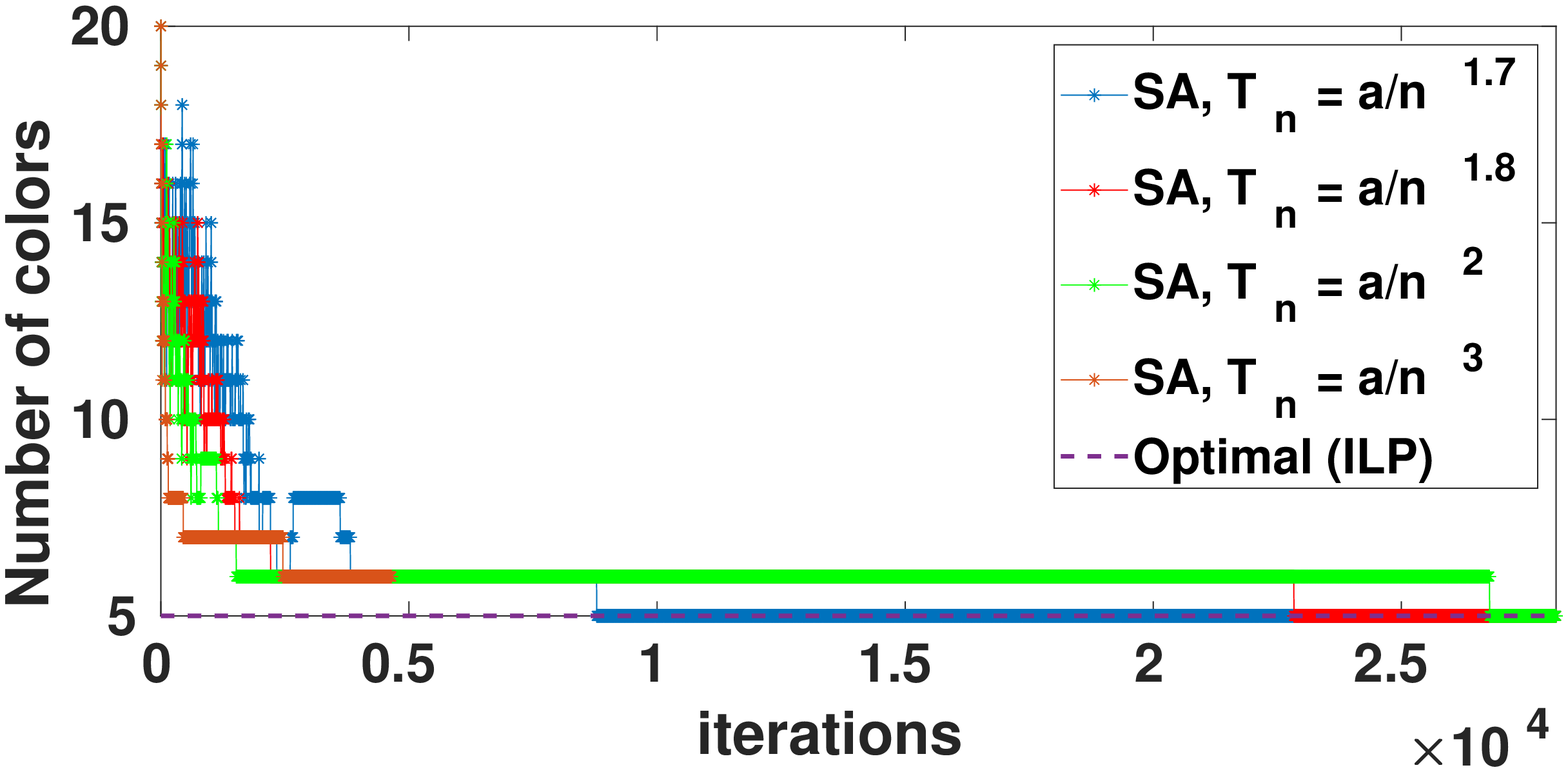} 
\includegraphics[scale=0.28]{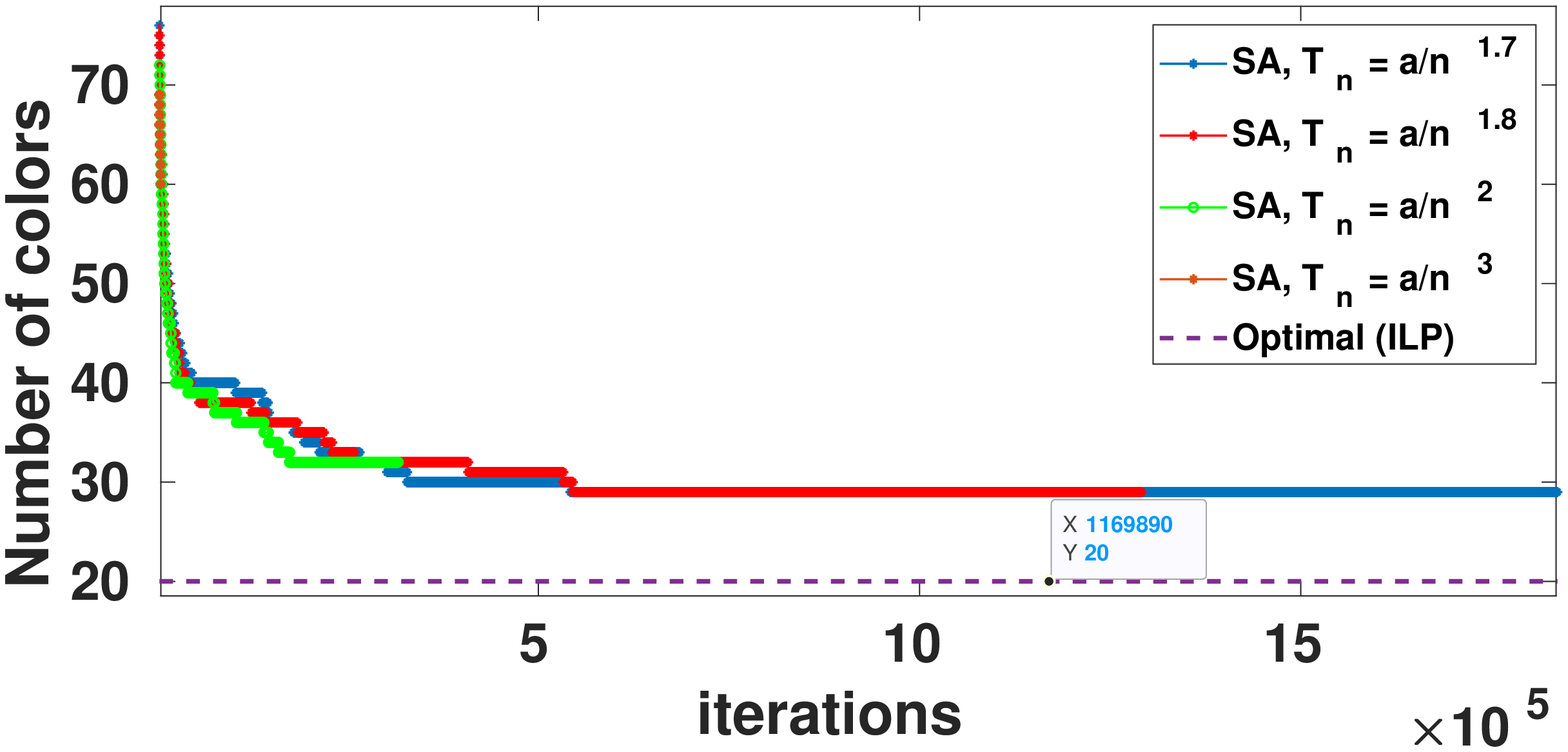}
\caption{Convergence SA for coloring of a)20 locations b)230 locations.}\label{coloring}
\end{figure}

\begin{figure}[t]
\centering
\includegraphics[scale=0.18]{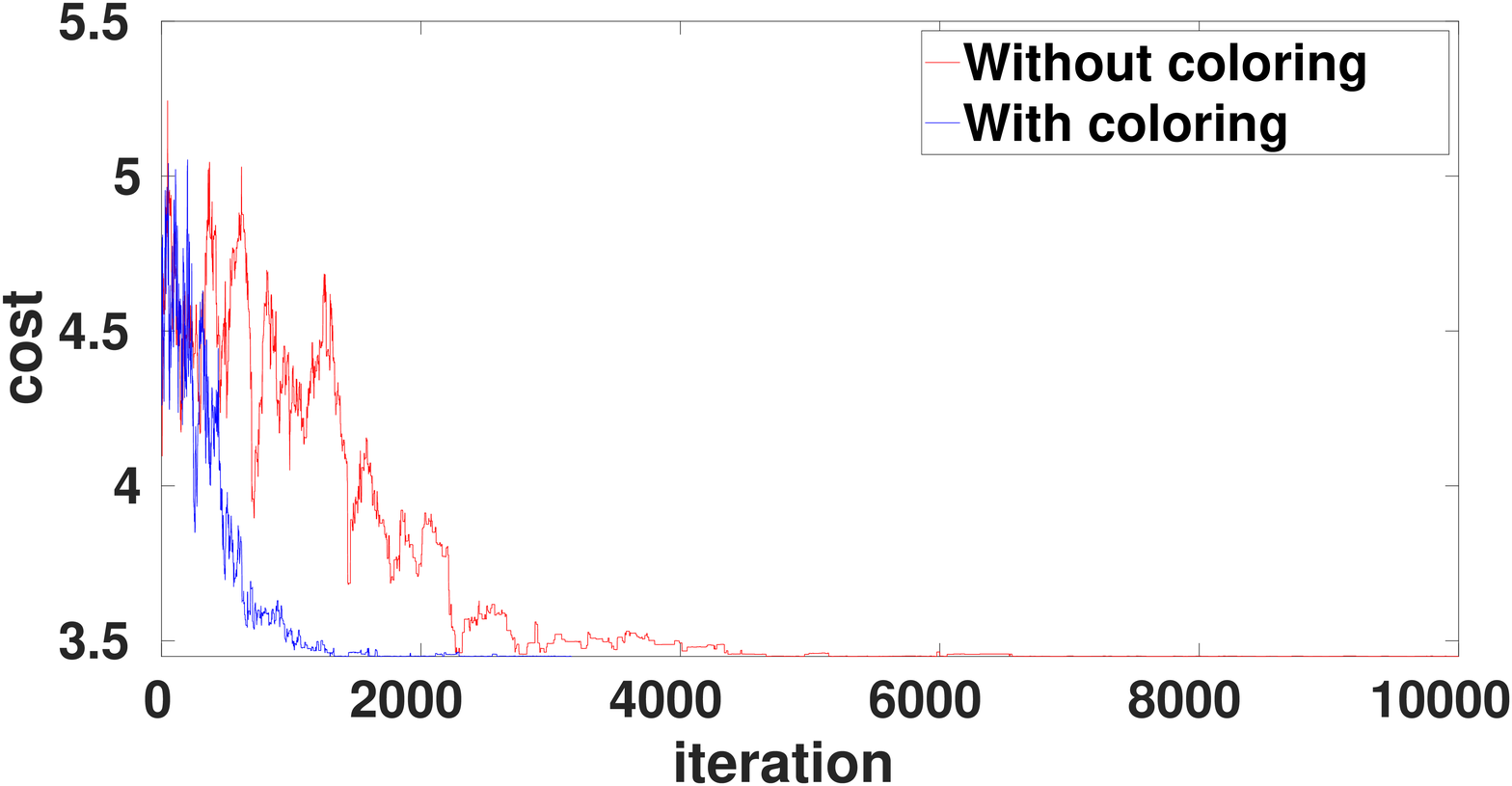}
\includegraphics[scale=0.18]{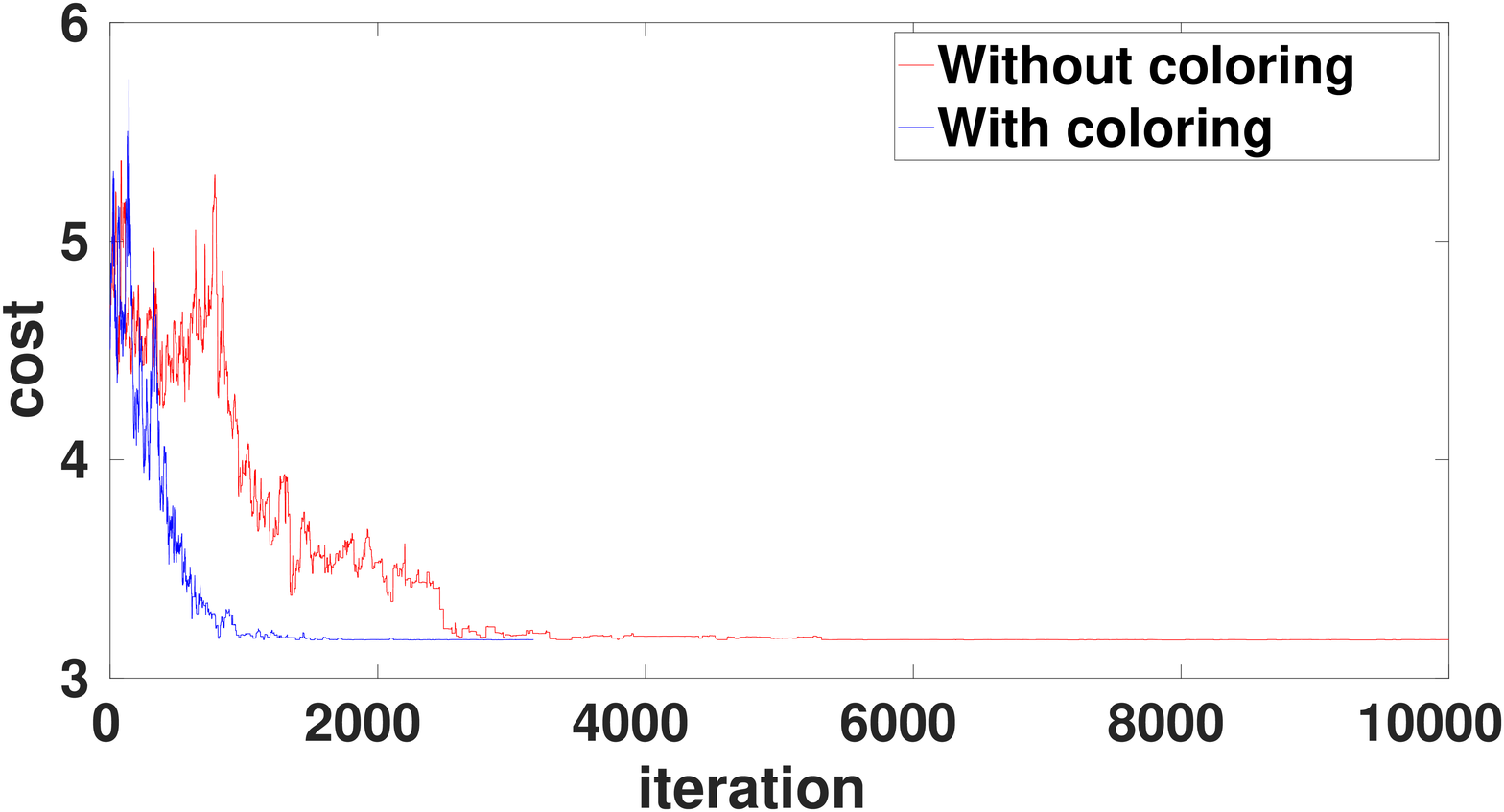} 
\includegraphics[scale=0.18]{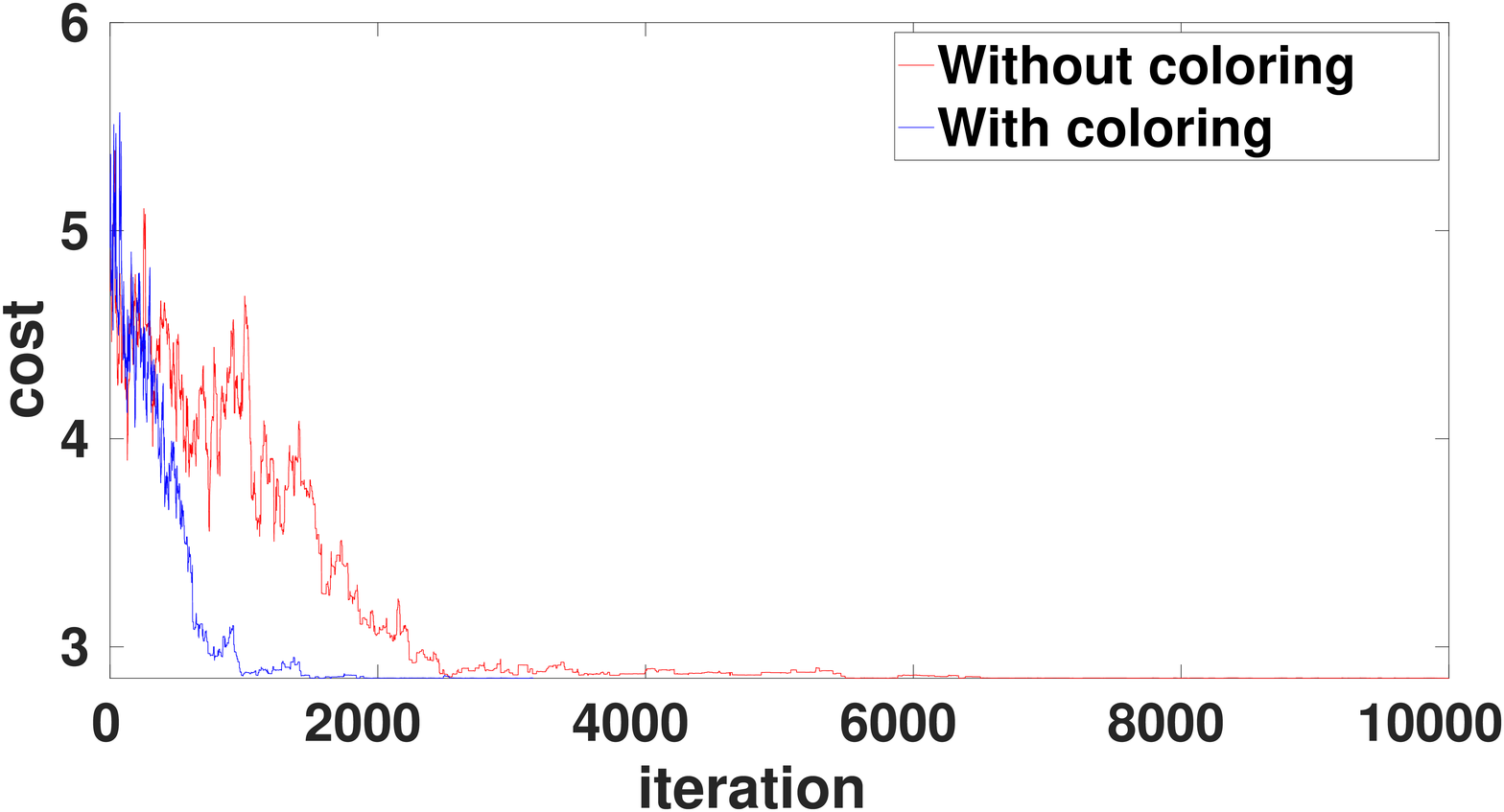} 
\caption {Comparing convergence of two algorithms for 1) $d = 7$ 2) $d = 9 $ 3) $d = 12$.} \label{fig:convergence}
\end{figure}

\section{Related work}\label{sec7}

Most related works tackle either the control of AoI in IoT networks or traffic offloading for 5G networks. \\
{\em Aging control for IoT}.  Aging control is at the core of IoT sensing applications, as it captures the trade-off between data staleness and resources 
utilization. Given the increasing demand  for IoT systems, the literature on control of AoI is correspondingly  growing. Most of the work on information 
aging control focus on users' standpoint, accounting for costs as perceived by the devices whose AoI is under control. Connections of AoI with traffic offloading are typically analyzed in the literature as a downstream effect of aging control. The potential relationship between AoI and traffic offloading has   been signaled in~\cite{altman2019forever}. With, previous works mostly focusing  on computation and task offloading~\cite{li2020age,liu2020joint,song2019age} rather than traffic offloading. In this paper, in contrast, we have considered jointly aging control and traffic offloading as first class citizens of an ecosystem wherein users and providers interact. In our scheme SPs influence users via pricing mechanisms able to couple aging control and traffic offloading in a unified framework.\\
{\em  Traffic offloading and slicing in 5G networks}. Utility service providers have long performed IoT data collection to reduce operational costs. Such traditional schemes are typically based on M2M to match the requirements of proprietary SCADA systems and charged per message. Nowadays, they appear inadequate for emerging IoT systems. In fact, the second major driver of the 5G technology, beyond multimedia traffic, is the current growth of mobile IoT connections \cite{cisco2020}. Actually, with both new LTE-M radio interface and the new suite of architectural paradigms, 5G interfaces key infrastructural assets able to ease both IoT access to radio resources and computing at the edge of the network. \\ Traffic offloading is enabled by 5G technology through slicing. Technical aspects such as slice insulation and fair slice allocation are still under development to upgrade LTE technology towards 5G, with large effort by the research community to overcome such technical issues \cite{Sciancalepore2017}\cite{WongTON2017}\cite{ZhengDeVeciana2018}. Nevertheless, slicing techniques are currently under standardization: specifications of the 5G system's slicing architecture and its requirements are  available \cite{3GPPPSlices}. In future 5G networks virtual private networks for IoT data collection will be shipped to SPs on top of the existing mobile network infrastructure with InP dedicated customer support. The traffic offloading mechanisms proposed in this work can be used by any such SP at the slice level for cost minimization purposes.

\section{conclusion}
Future IoT service providers will need ubiquitous IoT data collection, mandating in turn the support of IoT access 
at scale over the 5G infrastructure. At the same time, new schemes to control data generation and upload should allow to SPs to perform 
IoT data brokerage across diverse access resources made available by concurrent infrastructure providers 
at different costs, in the form of 5G IoT resource slices. 

This paper introduces a new framework to connect two fundamental aspects: the AoI of IoT data to be uploaded  
and the cost of 5G resources leased in order to obtain network access services. The upload control can be performed 
in a distributed way at the device level using optimal dynamic multi-threshold policies. Such policies have been showed to 
outperform their static counterparts. At same time, a SP can control \shadow prices to match optimal multi-threshold 
policies to service requirements while minimizing operational costs. It does so at the slice level by incentivizing users
 to perform IoT data uploads where resources leased from InPs are cheaper. This work opens new directions at the bridge between IoT and 5G research, by describing on a quantitative basis how to trade-off IoT data freshness and load balancing, as supported by the 5G slicing paradigm. In particular, we envision real testbed deployments and the investigation of  strategic mobility patterns to reduce costs as interesting areas for future exploration.

\bibliographystyle{IEEEtran}
\bibliography{literature,ageofinfo}
\appendix
\subsection{Proof of Theorem~\ref{THRESHOLD}} \label{ap:prooftheor1}

Next, we present the  proof of Thm.\ref{THRESHOLD}. 
 We begin by noting that from \cite{Puterman}  there exist a value function $V(x,l)$ and a scalar  $\rho$ satisfying the Bellman equation for the average cost MDP problem
\begin{align}
&  V(x,l)+\rho=\max \Big( U(x)-p(l) +  \sum_{l'\in \mathcal{L}} {\lambda}_{ll'} V(1,l'),  \nonumber \\ 
 &\qquad  U(x)  +  \sum_{l'\in \L} {\lambda}_{l{l{'}}} V(\min(x+1,M),l') \Big)  \label{eq:bellman}
\end{align}

An optimal policy $\mu$ able to select the per-state action maximising the right hand side of \eqref{eq:bellman} is an optimal solution to \eqref{eq:avcost}. 
Moreover, it is known that an unconstrained MDP admits a deterministic optimal policy \cite{Puterman}. Since a multi-threshold strategy belongs to this class of policies,  we restrict our discussion to the case of deterministic policies.

In what follows, we consider locations  sorted by increasing \shadow price order, that is $p(l) \leq p(l+1)$, for $l=1,\ldots,L$.  
Let us define the function $H : {\mathcal{M}} \times {\mathcal{L}} \times \{0, 1\}\rightarrow \mathbb R$ as follows 
\begin{align}
	H(x,l,1) &= U(x)-p(l)+\sum_{l'\in \mathcal L} {\lambda}_{ll'}V(1,l')\label{H1}\\
	H(x,l,0) &= U(x)+\sum_{l'\in \mathcal L} {\lambda}_{ll'}V(x+1,l')\label{H0}\\
	\Delta H(x,l) &= H(x,l,1) - H(x,l,0)
\end{align}

Hereafter we shall demonstrate that $i)$ the value function is decreasing in the age of information for any given location, 
$ii)$  that the optimal policy for any given location switches from $0$ to $1$ at most once and finally $iii)$ that if uploading is optimal for a certain value of the age of information at a given price, it is also optimal for larger prices as well. Such facts are proved formally in the following  lemma. 
\begin{lemma}\label{lem:monotone}
For any optimal policy, for $x=2,\ldots,M$ the following facts hold: \\
i. $V(x-1,l) \geq V(x,l),\; \forall  l\in{\mathcal L}$.\\
ii. $\Delta H(x-1,l) \geq 0 \Rightarrow \Delta H(x,l) \geq 0,\; \forall  l\in{\mathcal L}$.\\
iii. $\Delta H(x,l-1) \geq 0 \Rightarrow \Delta H(x,l)   \geq 0,\; \forall  l\in{\mathcal L}$.\\
iv. $V(x,l-1) \geq V(x,l),\; \forall  l\in{\mathcal L}.$\\
\end{lemma}
\begin{proof} We show each of the four items above in the corresponding order.

 $i.$ We can verify the result directly by backward induction on \eqref{eq:bellman}. For a deterministic policy, we define 
\begin{align}\label{eq:bellman2}
 Z_l(x) &:= V(x,l) - U(x)+\rho = \\
& =\max \Big( -p(l) + \sum_{l'\in \mathcal L} {\lambda}_{ll'}V(1,l'),\nonumber\\
& \quad\sum_{l'\in \mathcal L} {\lambda}_{ll'}V(\min(x+1,M),l') \Big)
\end{align}
We shall prove that $Z_l(x) \geq Z_l(x+1)$. This implies $V(x,l) \geq V(x+1,l)$ since $V(x,l) - U(x) \geq V(x+1,l) - U(x+1) \geq V(x+1,l) - U(x)$, where the last step holds because $U$ is non increasing. First, we observe that
\begin{align}
&Z(M-1)= \\
& \max \Big (-p(l) + \sum_{l'\in \mathcal L} {\lambda}_{ll'}V(1,l') , \sum_{l'\in \mathcal L} {\lambda}_{ll'}V(M,l') \Big )=Z_l(M)\nonumber 
\end{align}
so that the inductive basis holds true.

Now, in the general case we can observe that if the statement is true for $x+1$, that is $Z_l(x+1,l) \geq Z_l(x+2,l)$, it needs to hold for $x$ as well. Using the induction hypothesis, we have  $V(x+1,l) \geq V(x+2,l)$ and  thus 
\begin{eqnarray}\label{eq:zeta}
&&\hskip-7mm Z_l(x)=\max (-p(l) + \sum_{l'\in \mathcal L} {\lambda}_{ll'}V(1,l'),\sum_{l'\in \mathcal L} {\lambda}_{ll'}V(x+1,l'))\nonumber \\
&&\hskip-6mm \geq\max (-p(l) + \sum_{l'\in \mathcal L} {\lambda}_{ll'}V(1,l') , \sum_{l'\in \mathcal L} {\lambda}_{ll'}V(x+2,l'))\label{Vdec}\\
&&\hskip-6mm =Z_l(x+1)
\end{eqnarray}
which concludes the inductive step. \\

$ ii.$ It is sufficient to write $\Delta H(x-1,l) - \Delta H(x,l) =  \sum_{l'\in {\Loc}} \lambda_{ll'} ( V(x+1,l') - V(x,l') )\leq 0$. \\

$ iii.$  In this case, we can directly verify
\begin{align}
&\Delta H(x,l-1) - \Delta H(x,l) = -p(l-1) + p(l) \nonumber  \\ 
& \quad + \sum_{l'\in {\Loc}} (  \lambda_{(l-1)l'} - \lambda_{ll'} )( V(1,l') - V(x+1,l')) \nonumber   \\ 
&\quad\geq  \tilde{\kappa} \sum_{l'\in {\Loc}}(  \lambda_{(l-1)l'} - \lambda_{ll'} ) = 0 \nonumber  
\end{align}
where  
\begin{equation}
\tilde{\kappa}  = \sup_ {l'\in {\Loc}} \{V(1,l') - V(x+1,l')\}.
\end{equation}

$iv.$ Immediate since $p(l+1)\geq p(l)$.
\end{proof}

In what follows, we complete
the proof of Theorem \ref{THRESHOLD}.
\begin{proof} 
The proof of the  multi-threshold structure is a consequence of Lemma~\ref{lem:monotone}. In particular, let us define 
\[
\thresholdrange{l}:=\max\{x| \Delta H(x,l) < 0 \}, \; l=1,\ldots,L.
\]
so that in location $l$ it is optimal to upload for $x \geq \thresholdrange{l}$ for all prices $p \leq P_l$; also, from $ii.$, it follows that $\thresholdrange{1} \leq \thresholdrange{2} \ldots  \leq \thresholdrange{L}$. 

\end{proof}

\subsection{Proof of Theorem~\ref{PR-THRE}}
    \label{app:propo} 
    \begin{proof}
$(i)$ The following condition has to be satisfied for the optimal policy to be always using price $P_1=0$ :
	\begin{equation}
	\Delta H(x,l) < 0,\;\forall  l \in {\mathcal L} / {\mathcal L}_1, \;x=1,..,M.
	\label{27}
	\end{equation}
	From  (\ref{H1}) and  (\ref{H0}),  the condition (\ref{27}) yields
	\begin{eqnarray}
&&	\sum_{l'\in \mathcal L} {\lambda}_{ll'}(V(1,l')-V(x+1,l')) < p(l),\nonumber\\
&&\hspace{2cm} \forall  l \in {\mathcal L} / {\mathcal L}_1, x=1,..,M-1.\label{M1}
\end{eqnarray}
Since the value function $V$ is non-increasing,  the  conditions in (\ref{M1})  are  satisfied if and only if
\begin{multline}
	\sum_{l'\in {\mathcal L}_1} {\lambda}_{ll'}[V(1,l')-V(M,l')] + \\ \sum_{l'\notin {\mathcal L}_1} {\lambda}_{ll'}[V(1,l')-V(M,l')] < p(l),
\label{28}  
\end{multline}
$\forall  l \in {\L} / {\L}_1.$  Let us assume that the device cannot upload data at any location. $l\in {\mathcal L}_1$:  from~\eqref{eq:bellman} we have 
\begin{eqnarray}
V(x,l) - V(x-1, l) = U(x) - U(x-1),\; \forall l\in  {\mathcal L}_1\hspace{0.6cm}\label{L1}\\
V(x,l) - V(x-1, l) = U(M) - U(x-1),\; \forall  l\in {\mathcal L} / {\mathcal L}_1\label{L0}
\end{eqnarray}
(\ref{L1}) and (\ref{L0}) yield, respectively,
\begin{eqnarray}
V(1,l) - V(M, l) = U(1) - U(M),\; \forall l\in  {\mathcal L}_1\hspace{01.1cm}\label{L11}\\
	V(1,l) - V(M, l) = \sum_{x=1}^{M}(U(M) - U(x)),\; \forall l\in {\mathcal L} / {\mathcal L}\label{L00}
\end{eqnarray}
Plugging these  values of  $V(1,l) - V(M, l)$ into (\ref{28}) gives the condition (\ref{p1}). 

The derivations of $(ii)$ and $(iii)$ are similar to the above proof. 
\end{proof}

\subsection{Proof of Corrolary~ \ref{c2}}
    \label{pr:corro2}     
\begin{proof}
$(i)$ The following condition has to be satisfied for the optimal policy to never upload :
	\begin{equation}
	\Delta H(x,l) < 0,\;\forall  l \in {\mathcal L}  \;x=1,..,M.
	\label{C27}
	\end{equation}
	From  (\ref{H1}) and  (\ref{H0}),  the condition (\ref{27}) yields
	\begin{eqnarray}
&&	\sum_{l'\in \mathcal L} {\lambda}_{ll'}(V(1,l')-V(x+1,l')) < p(l),\nonumber\\
&&\hspace{2cm} \forall  l \in {\mathcal L} , x=1,..,M-1.\label{CM1}
\end{eqnarray}
Since the value function $V$ is non-increasing,  the  conditions in (\ref{CM1})  are  satisfied if and only if
\begin{multline}
	\sum_{l'\in {\mathcal L}} {\lambda}_{ll'}[V(1,l')-V(M,l')]  < p(l),
\label{C28}  
\end{multline}
$\forall  l \in {\L}$.  Let assume the device can not upload data at any location $l\in {\mathcal L}$:  from~\eqref{eq:bellman} we have 
\begin{eqnarray}
V(x,l) - V(x-1, l) = U(M) - U(x-1),\; \forall  l\in {\mathcal L}\label{CL0}
\end{eqnarray}
(\ref{CL0}) yields,
\begin{eqnarray}
	V(1,l) - V(M, l) = \sum_{x=1}^{M}(U(x) - U(M)),\; \forall l\in {\mathcal L} \label{CL00}
\end{eqnarray}
Plugging these  value of  $V(1,l) - V(M, l)$ into (\ref{C28}) gives the condition (\ref{c1}). 

\end{proof}

\subsection{Proof of Lemma~\ref{lem1}} \label{app:lemma1}

If the temperature is fixed, the transition matrix of the resulting time-reversible Markov chain, $Q_T=(q_T(\thresholdvector, \thresholdvector'))$, $0\leq \threshold_i,\threshold_i'\leq \threshold_{max}$, is given 
by: 
\begin{align} \label{eq:mattrans}
&q_T(\thresholdvector, \thresholdvector')= \\
& \quad\left\{
\begin{array}{ll}
\frac{q^*(\thresholdvector, \thresholdvector')  \pi_T(\thresholdvector')}{ \pi_T(\thresholdvector)}, &  \hspace{-2.5cm}\textrm{ if }  \frac{\pi_T(\thresholdvector')}{ \pi_T(\thresholdvector)}<1 \textrm{ and }  \thresholdvector\not=\thresholdvector'\\
q^*(\thresholdvector, \thresholdvector'),  & \hspace{-2.5cm} \textrm{ if }  \frac{\pi_T(\thresholdvector')}{ \pi_T(\thresholdvector)}\geq 1 \textrm{ and }  \thresholdvector\not=\thresholdvector'\\
q^*(\thresholdvector, \thresholdvector)  +\sum_{{\bf z}}q^*(\thresholdvector, {\bf z}) \left(1-\min\{1, \frac{\pi_T({\bf z})}{ \pi_T(\thresholdvector)}\}\right), &  \hspace{-0.3cm} \textrm{ if }  \thresholdvector=\thresholdvector'.
\end{array}
\right.
\end{align}
In what follows, we drop subscript $T$ to simplify presentation.

\begin{proof}
The above discrete time Markov chain, with transition probability matrix given by~\eqref{eq:mattrans}, has stationary distribution given by $\pi(\thresholdvector)$  if the following balance equations hold
\begin{equation}
\sum_{\thresholdvector \in S} \pi(\thresholdvector) q(\thresholdvector,\thresholdvector') = \pi(\thresholdvector'). \label{eq:balanc}
\end{equation}
Next, we show that the above equality holds. Indeed,
\begin{align}
&\sum_{\thresholdvector \in S} \pi(\thresholdvector) q(\thresholdvector,\thresholdvector') = \\
& \quad = \sum_{\thresholdvector \in C_1(\thresholdvector')}\pi(\thresholdvector) q(\thresholdvector,\thresholdvector') + \sum_{\thresholdvector \in C_2(\thresholdvector')}\pi(\thresholdvector) q(\thresholdvector,\thresholdvector')  + \nonumber \\ & \quad \quad \quad + \pi(\thresholdvector') q(\thresholdvector',\thresholdvector')
\end{align}
where
\begin{align}
C_1(\thresholdvector') &= \lbrace \thresholdvector  \in \thresholdvector_{\epsilon, d}  \setminus \{ \thresholdvector' \} \ \textbar \ \pi(\thresholdvector')<\pi(\thresholdvector)\rbrace \\ 
C_2(\thresholdvector') &= \lbrace \thresholdvector  \in \thresholdvector_{\epsilon, d}  \setminus \{ \thresholdvector' \} \ \textbar \ \pi(\thresholdvector')\geq\pi(\thresholdvector)\rbrace
\end{align}
Then,
\begin{align}
&\sum_{\thresholdvector \in S} \pi(\thresholdvector) q(\thresholdvector,\thresholdvector') =\\
 & = \sum_{\thresholdvector \in C_1(\thresholdvector')}\pi(\thresholdvector) \frac{q^*(\thresholdvector,\thresholdvector')\pi(\thresholdvector')}{\pi(\thresholdvector)} + \sum_{\thresholdvector \in C_2(\thresholdvector')}\pi(\thresholdvector) q^*(\thresholdvector,\thresholdvector') \nonumber  \\ 
 & \quad + \pi(\thresholdvector') \left( q^*(\thresholdvector',\thresholdvector')  +  \sum_{z \in C_2(\thresholdvector')}q^*(\thresholdvector',z)\left(1-\frac{\pi(z)}{\pi(\thresholdvector')}\right) \right)\\ 
& = \sum_{\thresholdvector \in C_1(\thresholdvector')}\pi(\thresholdvector')q^*(\thresholdvector,\thresholdvector') \nonumber  \\
& \quad + \cancel{\sum_{\thresholdvector \in C_2(\thresholdvector')}\pi(\thresholdvector) q^*(\thresholdvector,\thresholdvector')}  + \pi(\thresholdvector') q^*(\thresholdvector',\thresholdvector')  \nonumber \\
& \quad  + \pi(\thresholdvector')\sum_{z \in C_2(\thresholdvector')}q^*(\thresholdvector',z)  - \cancel{\sum_{z \in C_2(\thresholdvector')}q^*(\thresholdvector',z) \pi(z) }
\end{align}
Finally, as $Q^*$ is symmetric and stochastic,  $q^*(\thresholdvector,\thresholdvector')= q^*(\thresholdvector',\thresholdvector)$, 
\begin{align}
&\sum_{\thresholdvector \in S} \pi(\thresholdvector) q(\thresholdvector,\thresholdvector') =\\
& =\pi(\thresholdvector') \left( \sum_{\thresholdvector \in C_1(\thresholdvector')}q^*(\thresholdvector,\thresholdvector')  + \sum_{\thresholdvector \in C_2(\thresholdvector')}q^*(\thresholdvector,\thresholdvector') + q^*(\thresholdvector',\thresholdvector')\right) \\
& = \pi(\thresholdvector'),
\end{align}
which shows that~\eqref{eq:balanc} holds and concludes the proof.
\end{proof}

\end{document}